\documentclass[aps,prd,preprint,groupedaddress]{revtex4-1}
\usepackage[english]{babel}
\usepackage{amsmath,amsthm}
\usepackage{amsfonts}
\usepackage{upgreek}
\newcommand{\dif}{\mathrm{d}}
\newcommand{\bx}{\mathbf{x}}
\newcommand{\by}{\mathbf{y}}

\newcommand{\br}{\mathbf{r}}
\newcommand{\bn}{\mathbf{n}}
\newcommand{\ext}{\mathrm{ext}}
\DeclareMathOperator*{\FP}{FP}
\theoremstyle{remark}
\newtheorem{lem}{Lemma}
\numberwithin{equation}{section}
\renewcommand{\theequation}{\arabic{section}.\arabic{equation}}
\begin{document}
\title{Post-Newtonian approximation for isolated systems by matched asymptotic expansions  \\I.  General structure revisited}
\author{W.~G. Dixon}
\email{graham.dixon@ntlworld.com}
\affiliation{Churchill College, Cambridge CB3 0DS, U.K.}
\date{\today}
\begin{abstract}
  In recent years post-Newtonian approximations for isolated slowly-moving systems in general relativity have been studied by means of matched asymptotic expansions.  A paper by Poujade \& Blanchet in 2002 made great progress by effectively reducing the use of such expansions to an algorithmic form.  It gave systematic procedures for the development of both near-zone and far-zone asymptotic expansions, avoiding the divergent integrals which often bedevilled such methods, and
  showed that these two expansions could be made to match exactly, a result described there as somewhat remarkable.

  This paper revisits that work and shows that there is unfortunately an error in it which invalidates the results of the matching process as given therein.  The present paper identifies that error and shows how it may be corrected to give valid matching results.  The correction is presented in a redevelopment somewhat different from that of their paper.  This shows that far from being somewhat remarkable, it is in fact inevitable that the match is exact.
  It is indeed remarkable that they could carry both expansions to the point at which matching becomes possible, but if it can be done at all, then the match is necessarily exact.

  A companion paper will apply this asymptotic matching to a model problem developed by the present author in 1979 as a test bed for future developments in post-Newtonian approximations.  In this model, the correct near-zone expansion was obtained by approximation from an exact solution.  It will be shown that there is a discrepancy between this expansion and results from the original development of Poujade \& Blanchet but that the corrected development presented here reproduces the result of the model problem exactly.
\end{abstract}
%\pacs{04.25.Nx}
\maketitle
% ----------------------------------------------------------------
\section{Introduction}
In recent years post-Newtonian approximations for isolated slowly-moving systems in general relativity have been studied by means of matched asymptotic expansions, see the review by Blanchet \cite{B:2013} and many references contained therein.
Great progress in this direction was made by Poujade \& Blanchet \cite{PB:2002} in a paper of 2002, results from which form a basis for Section 5 of \cite{B:2013}.  That paper achieved what might well have been thought to be impossible.  This was a matching to all orders of approximation of a near-zone post-Newtonian asymptotic expansion with a corresponding far-zone post-Minkowskian one.  As a necessary prerequisite, itself no mean feat, procedures were given for generating both expansions to any desired order by means that avoided the divergent integrals which had often occurred in earlier work and which had cast doubt on the validity of the post-Newtonian approximation procedure itself.

I came across this paper only recently.  When I did so I decided to test their results against a model problem that I had developed \cite{Dixon:1979} in 1979 specifically as a test bed for slow-motion approximations in general relativity.  That model was sufficiently simple that it could be solved exactly by Fourier analysis in time and expansion in spherical harmonics.  The solution was expanded in powers of $1/c$ and the spherical harmonics re-summed to give a near-zone asymptotic expansion which included two features not at that time seen in any corresponding expansion in general relativity.  The first was that it was not a series purely in powers of $1/c$ but also included terms involving $\log(1/c)$.  The other was the appearance of tail terms involving integration over all past time.

Both these features are present in the results of \cite{PB:2002} so the model seemed to provide a worthwhile test.  I found, however, that the methods of \cite{PB:2002} gave additional terms not present in the model. These were traced to a subtle error in the matching process in \cite{PB:2002}.  The present paper identifies the error and shows how it may be corrected.  It is shown in a companion paper that the corrected version does reproduce the results of the model.

Presentation of the error and its correction necessarily requires some of the development of Poujade \& Blanchet to be repeated in a corrected form.  The opportunity has been taken to present a somewhat modified treatment that, it is hoped, sheds new light on their approach.  In particular it makes clear why, once matching becomes possible, it will inevitably be exact.  Any failure of the match to be exact would in fact indicate an error in the development of one or both of the expansions concerned.  One major result used by Poujade \& Blanchet in developing the far-zone expansion depends on an indirect proof based on a uniqueness argument.  To ensure that this does not hide any other subtle error, this paper gives a direct proof of the result concerned.

\section{Preliminaries and notation\label{prelim}}
The development in \cite{PB:2002} is set in a specific context and uses a specialized notation.  As these will be needed in what follows, they are summarised here. The context is a family of solutions of the Einstein field equations
\begin{equation}\label{2.1}
  G^{\mu\nu} = \kappa T^{\mu\nu}
\end{equation}
parameterized by the Einstein gravitational `constant' $\kappa$, treated as variable, with the spacetime being flat in the limit $\kappa\rightarrow 0$ and asymptotically flat for all $\kappa$.  Here $G^{\mu\nu}$ is the Einstein tensor of the spacetime metric and $T^{\mu\nu}$ is the stress-energy tensor of the material source.  The spacetime manifold $M$ and its coordinate system are considered to be fixed, with the components of the metric and stress-energy tensors being smooth functions of $\kappa$.

Lower-case Greek indices run from 0 to 3, lower-case Latin indices from 1 to 3, and the summation convention applies to repeated indices.  Upper-case indices have a special meaning described below. Partial differentiation with respect to the $\mu$th coordinate is denoted by $\partial_\mu$ and repeated differentiations by $\partial_{\mu\nu\ldots} := \partial_\mu\partial_\nu\ldots$.  Other than (\ref{2.1}), equations with Greek indices are not in tensor form and are relations between components rather than geometric objects. The sets of complex, real, integer and non-negative integer numbers are denoted respectively by $\mathbb{C}$, $\mathbb{R}$, $\mathbb{Z}$ and $\mathbb{N}$.

The components of the contravariant metric tensor density are denoted by $\mathfrak{g}^{\mu\nu}$.  The Einstein field equations are not used in the form (\ref{2.1}) but instead in the equivalent form
\begin{equation}\label{2.2}
  \partial_{\rho\sigma}(\mathfrak{g}^{\rho\sigma} \mathfrak{g}^{\mu\nu} - \mathfrak{g}^{\rho\mu} \mathfrak{g}^{\sigma\nu}) = 2\kappa(-g)(T^{\mu\nu} + t^{\mu\nu})
\end{equation}
where $t^{\mu\nu}$ is the Landau-Lifshitz gravitational pseudo-tensor \cite{LL:1951} and $g$ is the determinant of the covariant form of the metric tensor (or equivalently of $\mathfrak{g}^{\mu\nu}$).

The mapping of the metric tensor to the manifold is such that the coordinates are harmonic for all $\kappa$, so that $\partial_\mu \mathfrak{g}^{\mu\nu} = 0$, with $\mathfrak{g}^{\mu\nu}$ taking a diagonal form with components $(-1,+1,+1,+1)$ in the limit $\kappa\rightarrow 0$ and at spatial infinity for all $\kappa$. This limiting value is denoted by $\eta^{\mu\nu}$ and the gravitational field is represented by the difference
\begin{equation}\label{2.3}
  h^{\mu\nu} := \mathfrak{g}^{\mu\nu} - \eta^{\mu\nu}
\end{equation}
which therefore vanishes in those limits. Under these conditions (\ref{2.2}) can be put in the form
\begin{equation}\label{2.4}
  \Box h^{\mu\nu} = 2\kappa\tau^{\mu\nu}
\end{equation}
where $\Box := \eta^{\mu\nu}\partial_{\mu\nu}$ is the d'Alembertian operator of flat spacetime,
\begin{equation}\label{2.5}
  \tau^{\mu\nu} := (-g)T^{\mu\nu} + (2\kappa)^{-1}\Lambda^{\mu\nu}
\end{equation}
and
\begin{equation}\label{2.6}
  \Lambda^{\mu\nu} := 2\kappa(-g)t^{\mu\nu} - h^{\rho\sigma}\partial_{\rho\sigma}h^{\mu\nu}.
\end{equation}

The field $\Lambda^{\mu\nu}$ depends only on $h^{\mu\nu}$.  It may be considered as a function of $h^{\mu\nu}$, expandable as an infinite series in $h^{\mu\nu}$ and its derivatives with all terms being at least quadratic in these quantities.  The expression for $t^{\mu\nu}$ has an overall factor $(2\kappa)^{-1}$ which cancels the $2\kappa$ in the first term of (\ref{2.6}) so the appearance of a parameter in the expression for $\Lambda^{\mu\nu}$ is misleading.

The right hand side of \eqref{2.2} is divergence-free since the left hand side is identically so.  This is a physical requirement equivalent to the covariant conservation equation satisfied by $T^{\mu\nu}$.  It is shown in \cite{PB:2002} that this leads to the harmonic coordinate condition in the form $\partial_\mu h^{\mu\nu} = 0$ being automatically satisfied by the solution of \eqref{2.4} given by the procedures of that paper. That proof remains valid with the developments of the present paper and so is not repeated here. The present paper is therefore concerned solely with the solution of the field equations in the form (\ref{2.4}).  Moreover, all that is needed concerning the dependence of $\Lambda^{\mu\nu}$ on $h^{\mu\nu}$ is its quadratic nature described above, so its explicit form need not be given here.

If $(x^0,x^1,x^2,x^3)$ are the coordinates of a point $x \in M$ then for all values of $\kappa$, $(x^1,x^2,x^3)$ are treated as the components of a three-dimensional flat-space Cartesian vector $\bx$ and $t := x^0/c$ is taken as time.  The unit of time is considered to be variable, so making $c$, the speed of light, be a variable parameter.  With these conventions the Einstein gravitational constant $\kappa$ is related to the Newtonian one $G$ by $\kappa = 8\uppi G/c^4$.

The field equations are solved by matching two asymptotic expansions, both of which are expressed in terms of spherical harmonics represented as symmetric trace-free (STF) tensors.  Let $(r,\theta,\phi)$ be spherical polar coordinates and $\bn(\theta,\phi)$, with components $n_i$,  be a unit vector in the direction represented by these polar angles.  Then the STF product of $l$ copies of $\bn$, formed by subtracting from their product such multiples of the unit tensor as is required to make the result symmetric and trace-free, has $(2l+1)$ linearly independent components that are themselves functions of $(\theta,\phi)$.  These form a basis for the spherical harmonics of order $l$ that is equivalent to, though different from, the more usual $Y_l^m(\theta,\phi)$.

A product of multiple copies of any vector is denoted by the same kernel letter but with multiple indices.  A set of distinct indices is denoted in abbreviated form by a single upper-case index with the corresponding lower-case letter denoting the number of indices.  Putting these two conventions together, $n_L$ denotes the product of $l$ copies of the vector $n_i$.  Similarly $\partial_L$ denotes $l$ repeated differentiations with respect to the spatial coordinates.

In addition to the common use of round brackets $(~)$ around a set of indices to denote symmetrisation and square brackets $[~]$ to denote antisymmetrisation, diamond brackets $<~>$ are used to denote the STF part, \emph{i.e.} the symmetric part with the traces removed as described above.  If the STF part is taken of the entire set of indices of a tensor then it is alternatively denoted by `hatting' the kernel letter of the tensor.  The spherical harmonics of order $l$ can therefore be represented either by $n_{<L>}$ or more concisely by $\hat{n}_L$.  The solutions of the Laplace equation with this angular dependence are $r^l \hat{n}_L$ and $r^{-l-1} \hat{n}_L$.  The first of these can be written even more concisely as $\hat{x}_L$ where $\bx \equiv r\bn$ is the position vector.

Equation (\ref{2.4}) is solved with the various fields involved being treated as functions of $\bx$ and $t$. Series expansions are obtained around two different limits, $G\rightarrow 0$ with $c$ fixed and $(1/c)\rightarrow 0$ with $G$ fixed.  Both limits have $\kappa\rightarrow 0$ but the second limit is singular since for all $x \in M$, $t \rightarrow 0$ as $(1/c)\rightarrow 0$.  Now for a given source the values of both $G$ and $T^{\mu\nu}$ depend on the unit of time, the latter having the dimensions of an energy density.  It is $G/c^2$ and $T^{\mu\nu}/c^2$ that are independent of this unit, so for the material source not to vanish as $(1/c)\rightarrow 0$ it is necessary to have $T^{\mu\nu} = \mathcal{O}(c^2)$ in this limit for fixed $\bx$, $t$ and $G$.  Moreover, to maintain consideration of a generic point in which $|\bx|/x^0$ tends neither to zero nor infinity in the limit $(1/c) \to 0$, if $t = x^0/c$ is to remain finite and nonzero then so also must $|\bx|/c$.  Letting $(1/c) \to 0$ is therefore equivalent to $|\bx| \to 0$, making this expansion be an asymptotic expansion for the limit $r \to 0$, where $r := |\bx|$.  It is tempting to think of the limit $G \to 0$ as giving a corresponding asymptotic expansion for the limit $r \to \infty$, but it is not a singular limit and validity of the expansion is not limited to the far zone.  This fact is crucial to the procedures of Poujade \& Blanchet, as will be seen in Section \ref{match}.

\section{The post-Newtonian expansion\label{PN}}
The expansion about the Newtonian limit $(1/c)\rightarrow 0$ will be considered first, as it is here that the problem with the treatment in \cite{PB:2002} arises.  Equation (\ref{2.4}) can be rewritten as
\begin{equation}\label{3.1}
  \triangle h^{\mu\nu} = \frac{16\uppi G}{c^4}\,\tau^{\mu\nu} + c^{-2}\,\partial^2_t h^{\mu\nu}
\end{equation}
where $\triangle$ is the Laplace operator and $\partial_t := \partial/\partial t$.  The source $\tau^{\mu\nu}$ is formally expanded in a power series in $1/c$ as
\begin{equation}\label{3.2}
  \bar{\tau}^{\mu\nu}(\bx,t,c) = \sum_{n=-2}^\infty c^{-n}\,\underset{n}{\bar{\tau}}^{\mu\nu}(\bx,t,c)
\end{equation}
which starts with a power $-2$ since it was seen above that $T^{\mu\nu} = \mathcal{O}(c^2)$.  The corresponding expansion of the solution is
\begin{equation}\label{3.3}
  \bar{h}^{\mu\nu}(\bx,t,c) = \sum_{n=2}^\infty c^{-n}\,\underset{n}{\bar{h}}^{\mu\nu}(\bx,t,c),
\end{equation}
which starts at $n=+2$ in view of the $c^{-4}$ factor in \eqref{3.1}.  The overline on the symbols $\bar{\tau}^{\mu\nu}$ and $\bar{h}^{\mu\nu}$ indicates that they are post-Newtonian expansions.  An overline will also be used on expressions to denote that they are to be expanded in this way in powers of $1/c$.

Note that the individual terms in \eqref{3.2} and \eqref{3.3} are also allowed to depend on $c$.  This is to allow for the fact that the solution may not have an expansion that is purely in powers of $1/c$. It will be seen in Section \ref{match} that terms of the form $c^{-n} (\log c)^p$ are also required.  Such terms are to be grouped by the power of $c$ they include, so giving a $c$-dependence to the coefficients in the post-Newtonian expansions.

From here on this additional $c$-dependence will be left implicit and when two such expansions are compared, they will be matched by the powers of $c$ that appear explicitly. Substituting \eqref{3.2} and \eqref{3.3} into \eqref{3.1} therefore gives
\begin{equation}\label{3.4}
  \triangle \underset{n}{\bar{h}}^{\mu\nu} = 16\uppi G \underset{n-4}{\tau^{\mu\nu}} + \partial^2_t \underset{n-2}{\bar{h}^{\mu\nu}}
\end{equation}
for $n\geq 2$, with the convention $\underset{0}{\bar{h}}^{\mu\nu} = \underset{1}{\bar{h}}^{\mu\nu} = 0$.  This may be iterated to give
\begin{equation}\label{3.5}
  \triangle^{\lfloor n/2 \rfloor} \underset{n}{\bar{h}}^{\mu\nu} = 16\uppi G  \sum_{i=0}^{\lfloor n/2 \rfloor-1} \partial_t^{2i}\triangle^{\lfloor n/2 \rfloor -i-1} \underset{n-2i-4}{\bar{\tau}^{\mu\nu}}
\end{equation}
where $\lfloor x \rfloor$ denotes the floor of $x$, \emph{i.e.} the greatest integer not exceeding $x$, for any $x \in \mathbb{R}$.

Equation (\ref{3.5}) has the generic form
\begin{equation}\label{3.6}
  \triangle^{k+1}\bar{h} = \bar{\tau}
\end{equation}
with $k \in \mathbb{N}$.  If $\bar{\tau}$ is spatially bounded then a particular solution is given by
\begin{equation}\label{3.7}
  \bar{h}(\bx,t) = \triangle^{-k-1}[\bar{\tau}](\bx,t) := - \,\frac{1}{4\uppi}\,\int \dif^3 \by \frac{|\bx-\by|^{2k-1}}{(2k)!}\,\bar{\tau}(\by,t).
\end{equation}
To verify this, $k$ applications of $\triangle$ may be taken under the integral sign to leave the standard Poisson integral solution for the remaining  $\triangle$ operator.  It is important to note that (\ref{3.7}) defines the functional  $\triangle^{-k-1}$ as a single entity.  It is not the same as the $(k+1)$-fold iteration of $\triangle^{-1}$, which does not exist as iterations after the first give divergent integrals.

In \cite{PB:2002}, $\triangle^{-1}$ is extended by a regularization process to operands that diverge as $r \to 0$ no faster than a negative power of $r$ and as $r \to \infty$ no faster than a positive power of $r$.  Here this process needs to be applied to $\triangle^{-k-1}$ for a general $k \in \mathbb{N}$.  The treatment follows closely that of Appendix B of \cite{PB:2002} and so will be described only in outline.

The source function $\bar{\tau}(\bx,t)$ in (\ref{3.7}) is multiplied by a regularization factor $|\tilde{\bx}|^B$ where $\tilde{\bx} := \bx/r^0$, $r_0$ is a positive real parameter of the regularization process with $B \in \mathbb{C}$, and the range of integration is split into two parts, $|\by| < \mathcal{R}$ and $|\by| > \mathcal{R}$ for some $\mathcal{R} > 0$.  Both integrals, together with their formal derivative with respect to $B$, will converge for a suitable range of $\Re(B)$ and so will be analytic functions of $B$ within those ranges, though the ranges for the two integrals will not in general overlap.  The two results are extended by analytic continuation and added.  Their sum is defined on the overlap between the two regions of analytic continuation and is independent of the choice of $\mathcal{R}$ since the integral taken over the range between two different spheres $|\by| = \mathcal{R}$ is analytic for all $B$.  It is therefore a well-defined functional of $|\tilde{\bx}|^B \bar{\tau}(\bx,t)$ that will be denoted by $\triangle_B^{-k-1}$, a notation not used in \cite{PB:2002} but which is used here for clarity.  Note that the suffix $B$ denotes only the operations of splitting, analytic continuation and adding.  It follows that $\triangle_B^{-k-1}$ can be applied to any operands that are analytic functions of $B$ for which the two split integrals converge for suitable ranges of $B$.

Consider now the case when the range of definition of $\triangle_B^{-k-1}[\tilde{r}^B \bar{\tau}]$, where $r$ is the radial function, includes a deleted neighbourhood of $B=0$.  It may then be expanded in a Laurent series about $B=0$ with the coefficients being functionals of $\bar{\tau}$.  The finite part, \emph{i.e.} the coefficient of $B^0$, is the required extension of the integral in (\ref{3.7}).  For this and other functions of a complex variable $B$, this operation of selecting the finite part will be denoted by $\FP\limits_{B=0}$.  The extension of (\ref{3.7}) is then
\begin{equation}\label{3.8}
  \widetilde{\triangle^{-k-1}}[\bar{\tau}] := \FP\limits_{B=0}\triangle_B^{-k-1}[\tilde{r}^B \bar{\tau}].
\end{equation}

It is shown in Appendix \ref{PI} that
\begin{equation}\label{3.9}
  \triangle_B^{-1}[\hat{n}_L r^{B+a}] = \frac{\hat{n}_L r^{B+a+2}}{(B+a+2-l)(B+a+3+l)}
\end{equation}
for $a \in \mathbb{Z}$ and that
\begin{equation}\label{3.9a}
\triangle_B^{-k-1}[\hat{n}_L r^{B+a}] = (\triangle_B^{-1})^{k+1}[\hat{n}_L r^{B+a}].
\end{equation}
Here $\hat{n}_L$ denotes the spherical harmonics of order $l$ expressed as symmetric trace-free (STF) tensors as described in Section \ref{prelim}.  Note that this is a special case.  In general $\triangle_B^{-1}$ cannot be iterated as $\triangle_B^{-1}[\tilde{r}^B \bar{\tau}]$ will vanish at infinity no faster than $(1/r)$, whatever the value of $B$.  What is distinctive about this special case is that the overall factor $r^B$ remains at all stages of iteration.

This result enables us to demonstrate by a counterexample that in general $\widetilde{\triangle^{-k-1}}[\bar{\tau}] \neq (\widetilde{\triangle^{-1}})^{k+1}[\bar{\tau}]$.  It is here that the error in \cite{PB:2002} originates, for it is stated there in the context of its equation (3.9) that `it is not difficult to show' that these are equal.  With $a={}-l-2$, (\ref{3.9}) above gives
\begin{equation}\label{3.10}
  \triangle_B^{-1}[\hat{n}_L r^{B-l-2}] = \frac{\hat{n}_L r^{B-l}}{(B-2l)(B+1)}
\end{equation}
from which
\begin{equation}\label{3.11}
  \widetilde{\triangle^{-1}}[\hat{n}_L r^{-l-2}] = \frac{\hat{n}_L r^{-l}}{(-2l)}
\end{equation}
and hence
\begin{equation}\label{3.12}
  \triangle_B^{-1}[r^B\widetilde{\triangle^{-1}}[\hat{n}_L r^{-l-2}]] = \frac{\hat{n}_L r^{B-l+2}}{(-2l)(B-2l+2)(B+3)}.
\end{equation}
On the other hand
\begin{equation}\label{3.13}
  \triangle_B^{-2}[\hat{n}_L r^{B-l-2}] = \frac{\hat{n}_L r^{B-l+2}}{(B-2l)(B+1)(B-2l+2)(B+3)}.
\end{equation}
These are both analytic at $B=0$ if $l \neq 1$, but when $l=1$ they both have a simple pole at $B=0$. The finite parts, however, differ.  Since $r^B = \mathrm{e}^{B\log r}$ they give respectively
\begin{equation}\label{3.14}
  (\widetilde{\triangle^{-1}})^2[\hat{n}_L r^{-3}] = \frac{1}{36}\hat{n}_L r(2 - 6 \log \frac{r}{r_0})
\end{equation}
and
\begin{equation}\label{3.15}
  \widetilde{\triangle^{-2}}[\hat{n}_L r^{-3}] = \frac{1}{36}\hat{n}_L r(5 - 6 \log \frac{r}{r_0}).
\end{equation}
Since $\hat{n}_L$ for $l=1$ is just the vector $\mathbf{n}$, these show that
\begin{equation}\label{3.16}
  \widetilde{\triangle^{-2}}[ r^{-4}\br] = (\widetilde{\triangle^{-1}})^2[ r^{-4}\br] + \frac{1}{12}\,\br.
\end{equation}

Return now to (\ref{3.7}) applied to $\tilde{r}^B\bar{\tau}$.  It follows from this that
\begin{equation}\label{3.17}
  \triangle(\triangle_B^{-k-1}[\tilde{r}^B\bar{\tau}]) = \left\{\begin{array}{ll}
                                              \triangle_B^{-k}[\tilde{r}^B\bar{\tau}] & \mbox{for } k > 0 \vspace{3pt}\\
                                              \tilde{r}^B\bar{\tau} & \mbox{for } k = 0.
                                            \end{array}\right.
\end{equation}
The result for $k>0$ follows from applying $\triangle$ to each of the two split integrals for $\triangle^{-k-1}[\tilde{r}^B\bar{\tau}]$.  This gives the corresponding result for each part within the range of $\Re(B)$ for which it converges, which therefore holds also for their analytic continuations and so also for their sum.  For the case $k=0$ the two split integrals correspond to the standard Poisson integral applied respectively to the source functions $H(\mathcal{R} - r)\tilde{r}^B\bar{\tau}$ and $H(r - \mathcal{R})\tilde{r}^B\bar{\tau}$ where $H$ is the Heaviside step function.  Applying $\triangle$ to each returns these source functions, again initially within the range of $\Re(B)$ for which it converges but then extending throughout the domain of analytic continuation.  Their sum therefore gives the result for $k=0$.  There is no problem with points where $|\bx|=\mathcal{R}$ since the results follow for all $\mathcal{R}>0$.

If the operands are expanded about $B=0$ in their Laurent series, $\triangle$ acts on each series term by term.  The results can therefore be equated term by term with the corresponding expansions of the right hand side of (\ref{3.17}), so in particular the finite parts can be equated to give
\begin{equation}\label{3.18}
  \triangle(\widetilde{\triangle^{-k-1}}[\bar{\tau}]) = \left\{\begin{array}{ll}
                                              \widetilde{\triangle^{-k}}[\bar{\tau}] & \mbox{for } k > 0 \\
                                              \bar{\tau} & \mbox{for } k = 0.
                                            \end{array}\right.
\end{equation}
For the case $k=0$ the Laurent series for $\tilde{r}^B\bar{\tau}$ is just the Taylor series for $\mathrm{e}^{B\log \tilde{r}}\bar{\tau}$ whose finite part is just $\bar{\tau}$.

It follows from (\ref{3.18}) that a particular solution of (\ref{3.5}) is
\begin{equation}\label{3.19}
  (\underset{n}{\bar{h}}^{\mu\nu})_{\mathrm{part}} = 16\uppi G  \sum_{i=0}^\infty \partial_t^{2i}\widetilde{\triangle^{-i-1}}[\underset{n-2i-4}{\bar{\tau}^{\mu\nu}}]
\end{equation}
where $\underset{n}{\bar{\tau}}^{\mu\nu}$ is taken to be zero if $n<-2$.  This convention allows the summation formally to be taken to infinity as all terms with $i>\lfloor n/2 \rfloor-1$ are zero.  It can also be seen from (\ref{3.18}) that this would still be \emph{a} particular solution if $\widetilde{\triangle^{-i-1}}$ was replaced by $(\widetilde{\triangle^{-1}})^{i+1}$ but it would not be the \emph{same} particular solution.  The two would differ by a solution of the corresponding homogeneous equation.  It is this latter choice that is made in \cite{PB:2002} for although the notation $\widetilde{\triangle^{-i-1}}$ is used, it is \emph{defined} there to be $(\widetilde{\triangle^{-1}})^{i+1}$.  It will be seen in Section \ref{match} that the choice is not a free one as the matching with the far-zone asymptotic expansion requires the choice (\ref{3.19}).

The required solution of (\ref{3.5}) may differ from either of these particular solutions by a solution of the corresponding homogeneous equation.  The required addition cannot be determined from boundary conditions, for although the terms of the series (\ref{3.3}) must be finite at the origin, the series is asymptotic.  It diverges at spatial infinity and so no boundary condition can be applied there.  The addition must be the general solution of the homogeneous equation finite at the origin, with parameters to be determined later from the matching process.

The route to the general form is somewhat different from that in \cite{PB:2002} as that paper works directly with (\ref{3.4}).  Its particular solution is taken as
\begin{equation}\label{3.20}
  (\underset{n}{\bar{h}}^{\mu\nu})_\mathrm{part} = 16\uppi G \widetilde{\triangle^{-1}}[\underset{n-4}{\tau^{\mu\nu}}] + \partial^2_t \widetilde{\triangle^{-1}}[\underset{n-2}{\bar{h}^{\mu\nu}}]
\end{equation}
which is then iterated.  This results in the use of $(\widetilde{\triangle^{-1}})^{i+1}$ directly, not leaving any room for the use of an alternative operator.  To work from (\ref{3.5}) it is the homogeneous part of the solution that is found iteratively.

As seen in Section \ref{prelim}, the general solution of the Laplace equation $\triangle \bar{h} = 0$ finite at the origin is a linear combination of the spherical harmonic solutions $\hat{x}_L$.  Hence $\triangle^2 \bar{h} = 0$ can be solved in two stages.  First $\triangle \bar{h}$ must be a solution of the Laplace equation and so can be equated to such a linear combination.  That equation then has a particular solution consisting of a linear combination of terms of the form $\widetilde{\triangle^{-1}}[\hat{x}_L]$, to which must be added the Laplace solution consisting of terms of the form $\hat{x}_L$.  Proceeding in this way gives the general solution of (\ref{3.5}) as
\begin{equation}\label{3.21}
  \underset{n}{\bar{h}}^{\mu\nu} = 16\uppi G  \sum_{i=0}^\infty \partial_t^{2i}\widetilde{\triangle^{-i-1}}[\underset{n-2i-4}{\bar{\tau}^{\mu\nu}}]
  + \sum_{l=0}^\infty \sum_{i=0}^{\lfloor n/2\rfloor-1} \underset{n}{B}{}_{L,i}^{\mu\nu}(t) (\widetilde{\triangle^{-1}})^i[\hat{x}_L]
\end{equation}
for $n\geq 2$, where the functions $\underset{n}{B}{}_{L,i}^{\mu\nu}(t)$ are the parameters of the homogeneous solution.

These solutions must satisfy (\ref{3.4}), which requires
\begin{equation}\label{3.22}
  \underset{n}{B}{}_{L,i}^{\mu\nu} = \partial_t^2 \underset{n-2}{B}\,{}_{L,i-1}^{\mu\nu}
\end{equation}
for $n\geq 4$ and $1\leq i\leq \lfloor n/2\rfloor-1$.  This iterates to give
\begin{equation}\label{3.23}
  \underset{n}{B}{}_{L,i}^{\mu\nu} = \partial_t^{2i}\underset{n-2i}{B}{}_{L}^{\mu\nu}
\end{equation}
valid for $n\geq 2$ and $0\leq i\leq \lfloor n/2\rfloor-1$, where $\underset{n}{B}{}_L^{\mu\nu}:=\underset{n}{B}{}_{L,0}^{\mu\nu}$.  Hence (\ref{3.21}) can be put in the form
\begin{equation}\label{3.24}
  \underset{n}{\bar{h}}^{\mu\nu} = 16\uppi G  \sum_{i=0}^\infty \partial_t^{2i}\widetilde{\triangle^{-i-1}}[\underset{n-2i-4}{\bar{\tau}^{\mu\nu}}]
  + \sum_{l=0}^\infty \sum_{i=0}^\infty \partial_t^{2i} \underset{n-2i}{B}{}_{L}^{\mu\nu} (\widetilde{\triangle^{-1}})^i[\hat{x}_L]
\end{equation}
where $\underset{n}{B}{}_{L,i}^{\mu\nu}$ is taken to be zero if $n<2$, a convention that has enabled the summation over $i$ to be extended to infinity.

Equation (\ref{3.24}) may be formally re-summed by putting it back into (\ref{3.3}) to give
\begin{equation}\label{3.25}
  \bar{h}^{\mu\nu} = \frac{16\uppi G}{c^4}\widetilde{\mathcal{I}^{-1}}[\bar{\tau}^{\mu\nu}] + \sum_{l=0}^\infty \sum_{i=0}^\infty c^{-2i}\partial_t^{2i} B_L^{\mu\nu}(\widetilde{\triangle^{-1}})^i[\hat{x}_L]
\end{equation}
where
\begin{equation}\label{3.26}
  B_L^{\mu\nu}(t) := \sum_{n=2}^\infty c^{-n}\underset{n}{B}{}_L^{\mu\nu}(t)
\end{equation}
and the functional $\widetilde{\mathcal{I}^{-1}}$ is defined by
\begin{equation}\label{3.27}
  \widetilde{\mathcal{I}^{-1}} := \sum_{i=0}^\infty c^{-2i}\partial_t^{2i} \widetilde{\triangle^{-i-1}}.
\end{equation}

This functional is the `operator of instantaneous potentials' defined by (2.20) of \cite{PB:2002} but it differs from that of (3.10) of that paper, which uses $(\widetilde{\triangle^{-1}})^{i+1}$ in place of $\widetilde{\triangle^{-i-1}}$.  The notation here is of course that of the present paper, as the two functionals were not distinguished in \cite{PB:2002} in the belief that they were identical.  The important point for present purposes is that it was (3.10) of \cite{PB:2002} that was used there for the post-Newtonian expansion, whereas here that expansion has been developed using the form given by (2.20) of \cite{PB:2002}.  Note that both forms give valid expressions for $\bar{h}^{\mu\nu}$ but for the same solution they will give different values for the parameters $B_L^{\mu\nu}$.

The final transformation of (\ref{3.25}) follows that of \cite{PB:2002} precisely.  Since $\hat{x}_L = r^l \hat{n}_L$, it follows from (\ref{3.9}) that
\begin{equation}\label{3.28}
  (\widetilde{\triangle^{-1}})^i[\hat{x}_L] = \frac{(2l+1)!!}{(2i)!!(2l+2i+1)!!}\,r^{2i}\hat{x}_L.
\end{equation}
It also follows from Lemma \ref{lem6} of Appendix \ref{PI} that
\begin{equation}\label{3.29}
  \hat{\partial}_L r^{2l+2i} = \frac{(2l+2i)!!}{(2i)!!}\,r^{2i}\hat{x}_L
\end{equation}
and hence
\begin{equation}\label{3.30}
  (\widetilde{\triangle^{-1}})^i[\hat{x}_L] = \frac{(2l+1)!!}{(2l+2i+1)!}\,\hat{\partial}_L r^{2l+2i}
\end{equation}
which enables (\ref{3.25}) to be put in the form
\begin{equation}\label{3.31}
  \bar{h}^{\mu\nu} = \frac{16\uppi G}{c^4}\widetilde{\mathcal{I}^{-1}}[\bar{\tau}^{\mu\nu}] + \sum_{l=0}^\infty \hat{\partial}_L \left\{\frac{1}{r}\bar{A}^{\mu\nu}(r,t) \right\}
\end{equation}
where
\begin{equation}\label{3.32}
  \bar{A}^{\mu\nu}(r,t) =  \sum_{i=l}^\infty \frac{1}{(2i+1)!} \left(\frac{r}{c}\right)^{2i+1} \partial_t^{2i-2l}\left( (2l+1)!! c^{2l+1} B_L^{\mu\nu} \right).
\end{equation}

Let $A_L^{\mu\nu}(t)$ be any $(2l+1)$-fold antiderivative of $(-1)(2l+1)!! c^{2l+1} B_L^{\mu\nu}(t)$. Consider the expression
\begin{equation}\label{3.33}
  \left(\overline{A_L^{\mu\nu}(t-r/c) - A_L^{\mu\nu}(t+r/c)}\right)/2
\end{equation}
where the overline denotes the formal infinite post-Newtonian expansion in powers of $(1/c)$.  In the Taylor expansions of the two terms on the right of (\ref{3.33}), the terms with even powers of $(1/c)$ cancel.  The odd terms with powers of $(2l+1)$ and more give precisely $\bar{A}^{\mu\nu}(r,t)$.  If $\bar{A}^{\mu\nu}(r,t)$ in (\ref{3.31}) is replaced by (\ref{3.33}), the additional terms in the operand of $\hat{\partial}_L$ are therefore even powers of $r$ less than $2l$.  By \eqref{A8} $\hat{\partial}_L$ maps these terms to zero. The substitution therefore leaves (\ref{3.31}) unchanged, putting the general solution for $\bar{h}^{\mu\nu}$ in the form
\begin{equation}\label{3.34}
  \bar{h}^{\mu\nu} = \frac{16\uppi G}{c^4}\widetilde{\mathcal{I}^{-1}}[\bar{\tau}^{\mu\nu}] + \sum_{l=0}^\infty \hat{\partial}_L \left\{\frac{\overline{A_L^{\mu\nu}(t-r/c) - A_L^{\mu\nu}(t+r/c)}}{2r}\right\}.
\end{equation}

This is the final form of the solution of (\ref{3.1}) as a post-Newtonian expansion.  It is not an explicit solution, but instead an equation to be solved by iteration for $\bar{h}^{\mu\nu}$, commencing with $\bar{h}^{\mu\nu}=0$.  The functions $\bar{\tau}^{\mu\nu}(\bx,t)$ and  $A_L^{\mu\nu}(t)$ will in general depend on $\bar{h}^{\mu\nu}$, for which the value from the previous iteration is to be used.

\section{The post-Minkowskian expansion}
There is no error in the treatment in \cite{PB:2002} of the post-Minkowskian expansion.  However, as it too involves the operator of instantaneous potentials $\widetilde{\mathcal{I}^{-1}}$, it needs to be presented in summary here to demonstrate that it is consistent only with the definition (\ref{3.27}) of this paper and not with the alternative definition that was considered in \cite{PB:2002} to be equivalent.

The post-Minkowskian expansion is applied only to the field $h_\ext^{\mu\nu}$ outside the material source, where $T^{\mu\nu} = 0$.  A boundary condition can therefore be applied at infinity but not at the origin $r=0$.  In this region (\ref{2.4}) and (\ref{2.5}) give
\begin{equation}\label{4.1}
  \Box h_\ext^{\mu\nu} = \Lambda_\ext^{\mu\nu}.
\end{equation}
A condition of no incoming radiation is applied by brute force by requiring the source and its solution to be past-stationary, taken as meaning independent of $t$ for $t<-T$ for some constant $T$.  The expansion here is in powers of $G$, taken as
\begin{equation}\label{4.2}
  h_\ext^{\mu\nu} = \sum_{n=1}^\infty G^n h^{\mu\nu}_{(n)}
\end{equation}
with a corresponding expansion of $\Lambda_\ext^{\mu\nu}$ as
\begin{equation}\label{4.3}
  \Lambda_\ext^{\mu\nu} = \sum_{n=2}^\infty G^n \Lambda^{\mu\nu}_{(n)}.
\end{equation}
This starts at $n=2$ since $\Lambda^{\mu\nu}$ is at least quadratic in $h^{\mu\nu}$.  It follows that
\begin{equation}\label{4.4}
  \Box h^{\mu\nu}_{(n)} = \Lambda^{\mu\nu}_{(n)}
\end{equation}
for $n\geq 1$, with the convention that $\Lambda^{\mu\nu}_{(1)} = 0$.  Each term in (\ref{4.2}) is separately independent of $t$ for $t<-T$.

The equation for $n=1$ is just the vacuum wave equation.  It is shown in \cite{BD:1986} that the general solution of this equation that is both independent of time for $t<-T$ and is vanishing at spatial infinity (the condition of asymptotic flatness) is expressible as the spherical harmonic expansion
\begin{equation}\label{4.5}
  h^{\mu\nu}_{(1)}(\bx,t) = \sum_{l=0}^\infty \hat{\partial}_L \left( \frac{X^{\mu\nu}_{(1)L}(t-r/c)}{r}\right),
\end{equation}
where $r=|\bx|$, for some STF tensor functions $X^{\mu\nu}_{(1)L}(t)$ that are constant for $t<-T$. These functions describe, and are determined by, the field outside the source.  However, if they are known then (\ref{4.5}) can be regarded as the \emph{definition} of a function $h^{\mu\nu}_{(1)}(\bx,t)$ for all $r>0$, even within the region occupied by the source.  This in turn determines $\Lambda^{\mu\nu}_{(2)}(\bx,t)$ for all $r>0$ and so gives a meaning to (\ref{4.4}) for $n=2$ throughout this region.

If it converged then one particular solution for $h^{\mu\nu}_{(2)}(\bx,t)$ would be the retarded integral
\begin{equation}\label{4.6}
  h^{\mu\nu}_{(2)}(\bx,t) = \Box_{\mathrm{Ret}}^{-1}[\Lambda^{\mu\nu}_{(2)}](\bx,t) := -\frac{1}{4\uppi}\int \frac{\dif^3 \by}{|\bx-\by|}\,\Lambda^{\mu\nu}_{(2)}\big(\by, (t-|\bx-\by|/c)\big).
\end{equation}
There is no problem with convergence at infinity but there is at $r=0$ since the contribution to (\ref{4.5}) for a particular $l$ diverges as $r^{-l-1}$ as $r \rightarrow 0$.  To handle this, \cite{PB:2002} and \cite{BD:1986} suppose that $X^{\mu\nu}_{(1)L} = 0$ for sufficiently large $l$.  This has no practical effect on the results obtained as they are independent of this value, which can be arbitrarily large.  It means, however, that $h^{\mu\nu}_{(1)}$ and so also $\Lambda^{\mu\nu}_{(2)}$ diverge as $r \rightarrow 0$ no faster than some negative power of $r$.

This allows a similar regularization process to be performed to that in Section \ref{PN}, with the particular solution taken as
\begin{equation}\label{4.7}
  h^{\mu\nu}_{(2)} = \widetilde{\Box_{\mathrm{Ret}}^{-1}}[\Lambda^{\mu\nu}_{(2)}] := \FP\limits_{B=0}\Box_{B\,\mathrm{Ret}}^{-1}[\tilde{r}^B \Lambda^{\mu\nu}_{(2)}]
\end{equation}
where again $\tilde{r}=r/r_0$ and $r_0 >0$ is the regularization parameter.  Due to the use of the retarded integral, this is constant for $t<-T$.  The general solution is therefore given by adding the general solution of the homogeneous equation constant for $t<-T$, already seen to have the form (\ref{4.5}).  This can be continued to higher values of $n$, so giving
\begin{equation}\label{4.8}
  h^{\mu\nu}_{(n)}(\bx,t) = \widetilde{\Box_{\mathrm{Ret}}^{-1}}[\Lambda^{\mu\nu}_{(n)}](\bx,t)
  + \sum_{l=0}^\infty \hat{\partial}_L \left( \frac{X^{\mu\nu}_{(n)L}(t-r/c)}{r}\right)
\end{equation}
where for each $n$, $X^{\mu\nu}_{(n)L}$ is zero for sufficiently large $l$.  Due to the nonlinearity of the construction of $\Lambda^{\mu\nu}$, however, the value of $l$ beyond which this happens is a function of $n$ that tends to infinity with $n$.

Each of the functions $h^{\mu\nu}_{(n)}$ is defined throughout $r>0$, so (\ref{4.2}) and (\ref{4.3}) can now be used to extend the domain of definition of $h_\ext^{\mu\nu}$ and $\Lambda_\ext^{\mu\nu}$ throughout this region, even within the material source.  With this extension, (\ref{4.8}) can be formally re-summed over $n$ to give
\begin{equation}\label{4.11}
  h_\ext^{\mu\nu}(\bx,t) = \widetilde{\Box_{\mathrm{Ret}}^{-1}}[\Lambda_\ext^{\mu\nu}](\bx,t)
  + \sum_{l=0}^\infty \hat{\partial}_L \left( \frac{X^{\mu\nu}_L(t-r/c)}{r}\right)
\end{equation}
where
\begin{equation}\label{4.12}
X^{\mu\nu}_L(t) := \sum_{n=1}^\infty G^n X^{\mu\nu}_{(n)L}(t).
\end{equation}
This is the equivalent for the post-Minkowskian expansion of \eqref{3.34} for the post-Newtonian one.

\section{Multipole expansion\label{multi}}

Matching between the two expansions is achieved by expressing both in multipole form.  Throughout this paper, $\mathcal{M}$ will denote the operation of expansion in a multipole series. The operator is used with varying meaning in \cite{PB:2002}, which is a potentially misleading; see the note in Appendix \ref{PI} following \eqref{C3b}.

Multipole expansion of the post-Newtonian result \eqref{3.34} consists simply in applying $\mathcal{M}$ to those terms not already in multipole form and using \eqref{C3b} to give
\begin{equation}\label{5.1}
  \mathcal{M}(\bar{h}^{\mu\nu}) = \frac{16\uppi G}{c^4}\widetilde{\mathcal{I}^{-1}} [\mathcal{M}(\bar{\tau}^{\mu\nu})] + \sum_{l=0}^\infty \hat{\partial}_L \left\{\frac{\overline{A_L^{\mu\nu}(t-r/c) - A_L^{\mu\nu}(t+r/c)}}{2r}\right\}.
\end{equation}
That of the post-Minkowskian result \eqref{4.11} is substantially more complicated.  Let
\begin{equation}\label{5.2}
  \mathcal{M}(\Lambda_\ext^{\mu\nu})(\bx,t) = \sum_{l=0}^\infty \hat{n}_L \Lambda_L^{\mu\nu}(r,t)
\end{equation}
be the multipole expansion of $\Lambda_\ext^{\mu\nu}$ and from it construct the integrated quantities
\begin{equation}\label{5.3}
  R_{L,B}^{\mu\nu}(x,t) := \frac{c^4}{16\uppi G}\,x^l \int_0^x \dif y\, \frac{(x-y)^l}{l!}\,\tilde{y}^B\left(\frac{2}{y}\right)^{l-1}\Lambda_L^{\mu\nu}\left(y,t+\tfrac{y}{c}\right)
\end{equation}
and
\begin{equation}\label{5.4}
  R_L^{\mu\nu}(t) := \frac{c^4}{4G}\, (-1)^{l+1}\FP\limits_{B=0} \int_0^\infty \dif y\,\tilde{y}^B\,y^{l+2} \int_1^\infty \dif z\, \left(\frac{z^2-1}{2}\right)^l \Lambda_L^{\mu\nu}\left(y,t-\tfrac{yz}{c}\right).
\end{equation}
The numerical factors at the beginning are for agreement with \cite{PB:2002}.  The factors $\tilde{y}^B$ arise from regularization, where $\tilde{y}=y/r_0$ and $r_0$ is the regularization parameter.  Then it has been shown by Blanchet and Damour \cite{BD:1986} that
\begin{equation}\label{5.5}
  \mathcal{M}(\widetilde{\Box_{\mathrm{Ret}}^{-1}}[\Lambda_\ext^{\mu\nu}]) = \frac{16\uppi G}{c^4}\big(S_1^{\mu\nu} + S_2^{\mu\nu}\big)
\end{equation}
where
\begin{equation}\label{5.6}
  S_1^{\mu\nu}(\bx,t) := \FP\limits_{B=0} \int_{-r}^r \dif s\,\sum_{l=0}^\infty \hat{\partial}_L \left\{\frac{R_{L,B}^{\mu\nu}\left(\tfrac{s+r}{2},t-\tfrac{s}{c}\right)}{r}\right\}
\end{equation}
and
\begin{equation}\label{5.7}
  S_2^{\mu\nu}(\bx,t) := - \frac{1}{4\uppi} \sum_{l=0}^\infty \frac{(-1)^l}{l!} \hat{\partial}_L \left\{\frac{R^{\mu\nu}_L\left(t-\tfrac{r}{c}\right) -R^{\mu\nu}_L\left(t+\tfrac{r}{c}\right)}{2r}\right\}.
\end{equation}
The multipole expansion of \eqref{4.11} is therefore
\begin{equation}\label{5.8}
  M(h_\ext^{\mu\nu}) = \frac{16\uppi G}{c^4}\big(S_1^{\mu\nu} + S_2^{\mu\nu}\big) + S_3^{\mu\nu}
\end{equation}
where
\begin{equation}\label{5.9}
  S_3^{\mu\nu}(\bx,t) = \sum_{l=0}^\infty \hat{\partial}_L \left( \frac{X^{\mu\nu}_L(t-r/c)}{r}\right)
\end{equation}

\section{Matching\label{match}}
In addition to the stress-energy tensor $T^{\mu\nu}$ of the source, the two equations \eqref{5.1} and \eqref{5.8} involve the unknown multipole functions $A_L^{\mu\nu}(t)$ and $X_L^{\mu\nu}(t)$.  These  equations represent different expansions of the same one-parameter family of solutions of the field equations \eqref{2.4}, as described in Section \ref{prelim}.  The aim of matching is to determine the functions $A_L^{\mu\nu}(t)$ and $X_L^{\mu\nu}(t)$ from this fact, by comparing the two expansions in a common region of validity.

Equation \eqref{5.1} represents a post-Newtonian expansion of the solution both inside and outside the source as a series asymptotic in the limit $r \to 0$.  Equation \eqref{5.8}, on the other hand, represents a field defined \emph{throughout} $r>0$ that agrees with the required solution everywhere outside the source.  \emph{It can therefore itself be expanded about $r=0$ as a post-Newtonian expansion} and the two post-Newtonian expansions should be able to be put in identical forms in their common region of validity, namely outside the source, by a suitable choice of $A^{\mu\nu}(t)$ and $X^{\mu\nu}(t)$.  It should be no surprise that this is possible, as it is not the matching of two different expansions but of one expansion arrived at in two very different ways.

The presentation in \cite{PB:2002} discusses the matching as being between two asymptotic expansions, a post-Newtonian one valid in the limit $r \to 0$ and a post-Minkowskian one valid in the limit $r \to \infty$, so it becomes a surprise that the terms not involved in the matching, which contain much of the complication of the Einstein field equations, are the same in both expansions.  This appears to be a misreading of the situation.

Two steps are needed before the matching can be performed as described.  Equation \eqref{5.1} needs to be simplified to a form specific to the exterior of the material source and \eqref{5.8} needs to be given a post-Newtonian expansion.

The simplification of \eqref{5.1} uses the results of Appendix \ref{PI} and needs to take into account that $\tau^{\mu\nu}$ is given by \eqref{2.5}. Let $\mathcal{R}$, the radius at which the integration is split for the regularization process, be chosen so that the material source lies entirely within the sphere $r=\mathcal{R}$.  Following \eqref{C0}, let $\mathcal{M} (\bar{\tau}^{\mu\nu})$ be expanded as
\begin{equation}\label{6.1}
  \mathcal{M}(\bar{\tau}^{\mu\nu})(\bx,t) = \sum_{l=0}^\infty\hat{n}_L \bar{\sigma}_L^{\mu\nu}(r,t).
\end{equation}
The first iteration of \eqref{5.1} has $\Lambda^{\mu\nu} = 0$, giving $\bar{\sigma}_L^{\mu\nu}(r,t) = 0$ for $r>\mathcal{R}$.  Then \eqref{C3} shows that throughout the region $|\bx|>\mathcal{R}$, $\triangle^{-k-1}[M(\bar{\tau}^{\mu\nu})]$ will  be a sum of terms of the form $\hat{n}_L r^a f_L^{\mu\nu}(t)$ with $a \in \mathbb{Z}$ and $a < 2k$.  It follows from $\eqref{5.1}$ that in this region, the iteration results in a value for $\mathcal{M}(\bar{h}^{\mu\nu})$ of the same form, with $a$ bounded above at any post-Newtonian level.  This will therefore also hold for the value of $\mathcal{M}(\bar{\Lambda}^{\mu\nu})$, and hence that of $\mathcal{M}(\bar{\tau}^{\mu\nu})$, to be used in $|\bx|>\mathcal{R}$ for the next iteration.

It will be seen that subsequent iterations lead also to terms with powers of $\log r$.  This more general form can be expressed as
\begin{equation}\label{6.2}
  \mathcal{M}(\bar{\tau}_\ext^{\mu\nu})(\bx,t) = \sum_{l=0}^\infty\hat{n}_L \bar{\tau}_L^{\mu\nu}(r,t)
\end{equation}
with
\begin{equation}\label{6.3}
  \bar{\tau}_L^{\mu\nu}(r,t) = \sum_{a,p} r^a (\log r)^p f_{L,a,p}^{\mu\nu}(t)
\end{equation}
where $a \in \mathbb{Z}$ and $p \in \mathbb{N}$.  A distinction has been made between $\bar{\tau}^{\mu\nu}$ and $\bar{\tau}_\ext^{\mu\nu}$ since although they are identical in $|\bx|>\mathcal{R}$, the expression \eqref{6.2} gives a meaning to $\bar{\tau}_\ext^{\mu\nu}$ also for $0<|\bx|<\mathcal{R}$, where the two will differ.

The simplification of \eqref{5.1} consists in expressing $\widetilde{\mathcal{I}^{-1}} [\mathcal{M}(\bar{\tau}^{\mu\nu})]$ in terms of $\widetilde{\mathcal{I}^{-1}} [\mathcal{M}(\bar{\tau}_\ext^{\mu\nu})]$ for $|\bx| > \mathcal{R}$.  The reason that this is a simplification is that the latter is more easily evaluated with use of the results of Appendix \ref{PI}.

It follows from \eqref{C3} that
\begin{multline}\label{6.4}
   \triangle^{-k-1}[\tilde{r}^B \mathcal{M}(\bar{\tau}^{\mu\nu}- \bar{\tau}_\ext^{\mu\nu})](\bx,t)
   = - \sum_{l=0}^\infty \frac{\hat{n}_L}{(2k)!\,(2l+1)}\\ \times \sum_{i=0}^k a^{(k)}_{l,i} r^{2i-l-1} \int_0^\mathcal{R} \dif y\,\big(\bar{\sigma}_L^{\mu\nu}(y,t)-\bar{\tau}_L^{\mu\nu}(y,t)\big)\tilde{y}^B y^{l+2(k-i+1)},
\end{multline}
all other contributions being zero since $\bar{\sigma}_L^{\mu\nu}-\bar{\tau}_L^{\mu\nu}$ vanishes in $r>R$.  Now from \eqref{6.3},
\begin{multline}\label{6.5}
  \int_0^\mathcal{R} \dif y\,\bar{\tau}_L^{\mu\nu}(y,t)  \tilde{y}^B y^{l+2(k-i+1)}\\
  = \sum_{a,p} \frac{\partial^p}{\partial B^p}\Bigg\{ \frac{r_0^{-B}\mathcal{R}^{B+a+l+2(k-i)+3}} {B+a+l+2(k-i)+3} \Bigg\} f_{L,a,p}^{\mu\nu}(t) \\
  = - \int_\mathcal{R}^\infty \dif y\,\bar{\tau}_L^{\mu\nu}(y,t)  \tilde{y}^B y^{l+2(k-i+1)}
\end{multline}
where the equalities each hold for the range of $\Re(B)$ for which the integrals converge.  The two ranges do not overlap, but the extended ranges resulting from analytic continuation are identical.

It follows that in the analytic continuation of \eqref{6.4}, the first integral of \eqref{6.5} can be replaced by the second one. Since $\bar{\tau}_L^{\mu\nu} = \bar{\sigma}_L^{\mu\nu}$ in $r>\mathcal{R}$, this gives
\begin{multline}\label{6.6}
   \triangle_B^{-k-1}[\tilde{r}^B \mathcal{M}(\bar{\tau}^{\mu\nu}- \bar{\tau}_\ext^{\mu\nu})](\bx,t)
   = - \sum_{l=0}^\infty \frac{\hat{n}_L}{(2k)!\,(2l+1)}\\ \times \sum_{i=0}^k a^{(k)}_{l,i} r^{2i-l-1} \int_0^\infty \dif y\,\bar{\sigma}_L^{\mu\nu}(y,t)\tilde{y}^B y^{l+2(k-i+1)}
\end{multline}
where analytic continuation of the right hand side is implicit.  Now from \eqref{A8} and the reflection formula for the $\Gamma$-function
\begin{equation}\label{6.7}
  \hat{n}_L r^{2i-l-1} = (-1)^{l+i}\,2^{2i-l}\,i!\,\frac{\Gamma(\tfrac{1}{2})}{\Gamma(l-i+\tfrac{1}{2})}\, \hat{\partial}_L \left(\frac{r^{2i-1}}{(2i)!}\right)
\end{equation}
while from \eqref{6.1} and \eqref{A11}
\begin{equation}\label{6.8}
  \bar{\sigma}^{\mu\nu}_L(y,t) = \frac{(2l+1)!!}{4\uppi\,l!}\int \dif\Omega' \, \bar{\tau}^{\mu\nu}(y\bn',t)\,\hat{n}'_L.
\end{equation}
This can be used to convert the integral in \eqref{6.6} into one over the three-dimensional space of $\by \equiv y\bn'$.  Since $\hat{y}_L = y^l\, \hat{n}'_L$, it follows from \eqref{3.9} and \eqref{3.9a} that
\begin{equation}\label{6.9}
   y^{l+2(k-i)}\hat{n}'_L = 2^{l+2k-2i+1} \, \frac{(k-i)!}{(2l+1)!!}\, \frac{\Gamma(l+k-i+\tfrac{3}{2})}{\Gamma(\tfrac{1}{2})}\,
   \widetilde{\triangle^{i-k}}[\hat{y}_L]
\end{equation}
when $i<k$.  Putting this all together and using the value of $a^{(k)}_{l,i}$ from \eqref{A3b} gives
\begin{multline}\label{6.10}
  \triangle_B^{-k-1}[\tilde{r}^B \mathcal{M}(\bar{\tau}^{\mu\nu}- \bar{\tau}_\ext^{\mu\nu})](\bx,t)\\
  = - \frac{1}{4\uppi}\sum_{l=0}^\infty \frac{(-1)^l}{l!}\,\hat{\partial}_L \Bigg\{\frac{1}{r}\sum_{i=0}^k \frac{r^{2i}}{(2i)!}\int \dif^3 \by\, \bar{\tau}^{\mu\nu}(\by,t)\tilde{y}^B \widetilde{\triangle^{i-k}}[\hat{y}_L]\Bigg\}.
\end{multline}

Now take the finite part and use \eqref{3.27}.  Exchange the order of summations over $k$ from \eqref{3.27} and $i$ from \eqref{6.10}.  The summation over $i$ then becomes the even terms of a Taylor expansion in powers of $r/c$.  On denoting post-Newtonian expansion as usual by an overline, the result can be put in the form
\begin{equation}\label{6.11}
  \widetilde{\mathcal{I}^{-1}}[\mathcal{M}(\bar{\tau}^{\mu\nu}- \bar{\tau}_\ext^{\mu\nu})](\bx,t)
  = - \frac{1}{4\uppi}\sum_{l=0}^\infty \frac{(-1)^l}{l!}\,\hat{\partial}_L
  \Bigg\{\frac{\overline{F_L^{\mu\nu}(t-r/c)+F_L^{\mu\nu}(t+r/c)}}{2r}\Bigg\}
\end{equation}
where
\begin{equation}\label{6.12}
  F_L^{\mu\nu}(t) := \sum_{k=0}^\infty \frac{1}{c^{2k}}\FP\limits_{B=0} \int \dif^3 \by\,\tilde{y}^B \widetilde{\triangle^{-k}}[\hat{y}_L]\, \partial_t^{2k}\bar{\tau}^{\mu\nu}(\by,t).
\end{equation}
This is equivalent to equation (C15) of \cite{PB:2002} but the derivation is of necessity different.  The derivation in \cite{PB:2002} is iterative and relies on the interpretation of $\widetilde{\triangle^{-k-1}}$ as $(\widetilde{\triangle^{-1}})^{k+1}$, which has been shown here to be flawed.  Equation \eqref{6.12} finally enables \eqref{5.1}, in a region $|\bx|>\mathcal{R}$ outside the source, to be expressed as
\begin{multline}\label{6.13}
  \mathcal{M}(\bar{h}_\ext^{\mu\nu}) = \frac{16\uppi G}{c^4}\widetilde{\mathcal{I}^{-1}} [\mathcal{M}(\bar{\tau}_\ext^{\mu\nu})]\\
   - \frac{4G}{c^4}\sum_{l=0}^\infty \frac{(-1)^l}{l!}\,\hat{\partial}_L
  \Bigg\{\frac{\overline{F_L^{\mu\nu}(t-r/c)+F_L^{\mu\nu}(t+r/c)}}{2r}\Bigg\} \\
    +\sum_{l=0}^\infty \hat{\partial}_L \left\{\frac{\overline{A_L^{\mu\nu}(t-r/c) - A_L^{\mu\nu}(t+r/c)}}{2r}\right\}.
\end{multline}
Since this can be solved throughout $r>0$, its solution has been denoted by $\bar{h}_\ext^{\mu\nu}$ to indicate that it differs from $\bar{h}^{\mu\nu}$ within the source.

To complete the proof of \eqref{6.13}, it remains to verify that at each iteration the input value of $\mathcal{M}(\bar{\tau}_\ext^{\mu\nu})$ does have the form given by \eqref{6.2} and \eqref{6.3}.  It has already been shown that this is true of the second iteration, with the terms in $\log r$ being absent.  It was seen in Section \ref{PN} that the action of $\widetilde{\triangle^{-k-1}}$ on such terms can generate terms with the $\log r$ factor.  The other terms in $\mathcal{M}(\bar{h}^{\mu\nu})$, arising from the post-Newtonian expansions of $F_L^{\mu\nu}$ and $A_L^{\mu\nu}$, will have only integer powers of $r$.  However, overall, this will lead to an input value of $\mathcal{M} (\bar{\tau}_\ext^{\mu\nu})$ for the next iteration of the form \eqref{6.2} and \eqref{6.3} that does include powers of $\log r$.  Further action of $\widetilde{\triangle^{-k-1}}$ may increase the power of $\log r$ but still preserve these forms, which are therefore the most general forms that can occur.

This completes the treatment of the post-Newtonian expansion so it remains to consider the post-Minkowskian one.  As described above, what is needed here is to expand \eqref{5.8} itself in  post-Newtonian form.  For the terms $S_2^{\mu\nu}$ and $S_3^{\mu\nu}$ of \eqref{5.7} and \eqref{5.9}, this consists simply of expanding them in powers of $r/c$, which will be denoted once again by an overline. It is only $S_1^{\mu\nu}$ that needs special treatment.

The expansion of this term has been obtained in both \cite{PB:2002} and \cite{B:1993} but both treatments, especially the former, involve an indirect step in which two solutions of a wave equation (in the case of \cite{PB:2002}) or an iterated Poisson equation (in the case of \cite{B:1993}) are identified as equal through both being proportional to the regularization factor $r^B$. The surprisingly simple result is that
\begin{equation}\label{6.14}
  \frac{16\uppi G}{c^4}\,\overline{S_1^{\mu\nu}} = \widetilde{\mathcal{I}^{-1}}[\overline{\mathcal{M}(\Lambda_\ext^{\mu\nu})}].
\end{equation}
Since the result involves the operator of instantaneous potentials that lies at the centre of the error in \cite{PB:2002}, to avoid all possible shadow of doubt a direct derivation of (\ref{6.14}) is given in Appendix \ref{PNOS}.  With this result, the complete post-Newtonian expansion of the post-Minkowskian solution takes the form
\begin{multline}\label{6.15}
  \overline{\mathcal{M}(h_\ext^{\mu\nu})} = \widetilde{\mathcal{I}^{-1}}[\overline{\mathcal{M}(\Lambda_\ext^{\mu\nu})}] \\
  - \frac{4G}{c^4} \sum_{l=0}^\infty \frac{(-1)^l}{l!} \hat{\partial}_L \left\{\frac{\overline{R^{\mu\nu}_L(t-r/c) -R^{\mu\nu}_L(t+r/c)}}{2r}\right\} \\
  + \sum_{l=0}^\infty \hat{\partial}_L \left( \frac{\overline{X^{\mu\nu}_L(t-r/c)}}{r}\right).
\end{multline}

Both \eqref{6.13} and \eqref{6.15} are expansions that are asymptotic in the limit $r \to 0$, valid only outside the source, but in this region they are alternative forms for the same expansion and so should be identical.  This can be expressed simply as the requirement
\begin{equation}\label{6.16}
  \mathcal{M}(\bar{h}_\ext^{\mu\nu}) = \overline{\mathcal{M}(h_\ext^{\mu\nu})}
\end{equation}
which is referred to in \cite{PB:2002} as the matching condition.

It follows from \eqref{2.5} that outside the source,
\begin{equation}\label{6.17}
  \overline{\mathcal{M}(\Lambda_\ext^{\mu\nu})} = \frac{16\uppi G}{c^4}\, \mathcal{M}(\bar{\tau}_\ext^{\mu\nu})
\end{equation}
as it is immaterial whether the multipole or the post-Newtonian expansion is performed first.  It is important, however, that this equality holds throughout $r>0$, even within the source.  This is required since $\widetilde{\mathcal{I}^{-1}}$ is a functional, not a function.  Indeed, if this were not the case then both sides of \eqref{6.11} would be zero outside the source.

The definitions of both sides of \eqref{6.17} have been extended to the whole of the region $r>0$, but in substantially different ways.  On the right it is extended by adopting \eqref{6.3} throughout the region, where the coefficients $f_{L,a,p}^{\mu\nu}(t)$ retain the values valid outside the source.  On the left it is extended by taking \eqref{4.5} to hold throughout $r>0$, where $X_{(1)L}^{\mu\nu}(t)$ is determined from the field outside the source.  It is shown by \eqref{B30} of  Appendix \ref{PNOS}, however, that the left hand side also has the form given by \eqref{6.2} and \eqref{6.3}, so both extensions are equivalent.  Hence \eqref{6.17} holds throughout $r>0$, as required.

The first terms on the right of each of \eqref{6.13} and \eqref{6.15} are therefore identical.  Matching then requires the remaining terms to be so as well, giving
\begin{equation}\label{6.18}
  A_L^{\mu\nu}(t) = - \frac{4G}{c^4}\,\frac{(-1)^l}{l!}\big(F_L^{\mu\nu}(t)+R_L^{\mu\nu}(t)\big),
\end{equation}
\begin{equation}\label{6.19}
  X_L^{\mu\nu}(t) = - \frac{4G}{c^4}\,\frac{(-1)^l}{l!}\,F_L^{\mu\nu}(t).
\end{equation}
This is in apparent agreement with (4.5) and (4.6) of \cite{PB:2002}, where the results are obtained from comparison of equations (4.2) and (4.3) of that paper.  However, those two equations were derived with differing definitions of $\widetilde{\mathcal{I}^{-1}}$ from one another, corresponding to equations \eqref{3.14} and \eqref{3.15} of the present paper.  This results in those two terms having values that differ by a solution of the homogeneous equation, so the matching results presented there are invalid.

Note that the matching process requires the same value of the regularization parameter $r_0$ to be used throughout.  The value used in the post-Minkowskian equation \eqref{4.11} becomes that used in \eqref{5.3} and \eqref{5.4}, in particular in $R_L^{\mu\nu}(t)$, and as seen in Appendix \ref{PNOS}, also in the first term on the right of \eqref{6.15}.  Matching requires this to be used in all the operations that comprise $\widetilde{\mathcal{I}^{-1}}$ in \eqref{6.13} and the development in Section \ref{match} shows this also to be the value used in $F_L^{\mu\nu}(t)$.  These individual contributions will in general all depend on $r_0$, but this dependence will cancel to give a final solution for $h^{\mu\nu}$ that is independent of the choice made.

The functions on the right of \eqref{6.18} and \eqref{6.19} are determined by \eqref{6.12} and \eqref{5.4} in terms of the field and source variables.  With these results, the two equations \eqref{6.13} and \eqref{6.15} therefore unify into a single equation that has no indeterminate functions and which can be solved by iteration.  There remains one consistency check to be performed, however.  This is to verify that the solution it yields does have the form \eqref{3.3} assumed at the outset.

The only question about this concerns the function $R_L^{\mu\nu}(t)$, as all other contributions have been seen to have expansions that are in powers of $1/c$. Now it follows from \eqref{B30} and \eqref{5.2} that the post-Newtonian expansion $\bar{\Lambda}_L^{\mu\nu}$ of $\Lambda_L^{\mu\nu}$ is a sum of terms of the form $r^a (\log r)^p\, G^{\mu\nu}_{L,a,p}(t)$, so consider the contribution to the corresponding expansion of $R_L^{\mu\nu}(t)$ from one such term.  As usual the result for $p>0$ can be obtained from that for $p=0$ by differentiating with respect to the regularization variable $B$, so take $\bar{\Lambda}_L^{\mu\nu}(r,t) = r^a G^{\mu\nu}_L(t)$.  Then \eqref{5.4} gives
\begin{equation}\label{6.20}
  \overline{R_L^{\mu\nu}}(t) = \frac{c^4}{4G}\, (-1)^{l+1}\FP\limits_{B=0} \int_0^\infty \dif y\,\tilde{y}^B\,y^{l+a+2} \int_1^\infty \dif z\, \left(\frac{z^2-1}{2}\right)^l G_L^{\mu\nu}(t-\tfrac{yz}{c}).
\end{equation}
If the variables of integration are changed from $(y,z)$ to $(u,v)$ with $y=cuv$, $z=v^{-1}$, the result factorizes to give
\begin{multline}\label{6.21}
  \overline{R_L^{\mu\nu}}(t) = \frac{c^{7+a+l}}{4G}\,(-1)^{l+1}\FP\limits_{B=0}\left(\frac{c}{r_0}\right)^B \bigg\{ \int_0^\infty \dif u\,u^{B+a+2+l} G_L^{\mu\nu}(t-u)\bigg\} \\
  \times \bigg\{\int_0^1 \dif v\, v^{B+a+1-l}\left(\frac{1-v^2}{2}\right)^l \bigg\}.
\end{multline}

It is only derivatives of $R_L^{\mu\nu}$ that occur in \eqref{6.15} as the first, and all other even-order, terms in its Taylor expansion cancel.  Since $G_L^{\mu\nu}(t)$ is constant for sufficiently large negative $t$ and only its derivative is relevant, there is no convergence problem in the integral limit $u \to \infty$.  Both integrals converge at their lower limit for sufficiently large $\Re(B)$, but considered as functions of $B$ their analytic continuations will in general have a pole at $B=0$.  The factor of $c^B$ will therefore give rise to factors of $\log c$ in the finite part.  The highest integer power of $1/c$ is not clear from \eqref{6.21}, but \eqref{5.4} shows it to be $-2$ since that for $\Lambda^{\mu\nu}$ is $+2$.  Equation \eqref{6.15} therefore shows that $\overline{R_L^{\mu\nu}}$ gives contributions to $\bar{h}^{\mu\nu}$ whose $c$-dependence has the form $c^{-n}(\log c)^p$ with $n \geq 2$, $p \geq 0$.  This is indeed consistent with \eqref{3.3} as required, but it shows that the terms in that expansion may themselves have a $c$-dependence through $\log c$, as stated there.

\section{Conclusions}
The use of matched asymptotic expansions by Poujade \& Blanchet in \cite{PB:2002} has been revisited.  This has revealed a subtle error in the matching process of that paper.  The results affected by the error have been redeveloped in a corrected form.  This redevelopment has identified why the matching achieved is necessarily exact, a result described in \cite{PB:2002} as non-trivial and somewhat remarkable.  The work of that paper is indeed remarkable in that it has carried both asymptotic expansions to the point at which they can be compared and therefore matched, but once that can be done, they must match exactly.  The redevelopment has also presented a direct proof of a major result from \cite{B:1993}, used and re-derived in \cite{PB:2002}.  Both of these preceding derivations have indirect steps based on a uniqueness argument.  Having a direct proof eliminates the possibility of any further subtle error hidden by this reliance on uniqueness.

The problem considered here and in \cite{PB:2002} is the post-Newtonian approximation to the solution of the Einstein field equations in harmonic coordinates for a bounded source.  The final outcome of the matching procedure is that this approximation can be obtained by iteration of the implicit equation
\begin{equation}\label{7.1}
  \bar{h}^{\mu\nu} = \frac{16\uppi G}{c^4}\widetilde{\mathcal{I}^{-1}}[\bar{\tau}^{\mu\nu}] + \sum_{l=0}^\infty \hat{\partial}_L \left\{\frac{\overline{A_L^{\mu\nu}(t-r/c) - A_L^{\mu\nu}(t+r/c)}}{2r}\right\}
\end{equation}
with
\begin{equation}\label{7.2}
  A_L^{\mu\nu}(t) = - \frac{4G}{c^4}\,\frac{(-1)^l}{l!}\big(F_L^{\mu\nu}(t)+R_L^{\mu\nu}(t)\big),
\end{equation}
\begin{equation}\label{7.3}
  F_L^{\mu\nu}(t) := \sum_{k=0}^\infty \frac{1}{c^{2k}}\FP\limits_{B=0} \int \dif^3 \by\,\tilde{y}^B \widetilde{\triangle^{-k}}[\hat{y}_L]\, \partial_t^{2k}\bar{\tau}^{\mu\nu}(\by,t),
\end{equation}
\begin{equation}\label{7.4}
  R_L^{\mu\nu}(t) := 4\uppi (-1)^{l+1}\FP\limits_{B=0} \int_0^\infty \dif y\,\tilde{y}^B\,y^{l+2} \int_1^\infty \dif z\, \left(\frac{z^2-1}{2}\right)^l \bar{\tau}_L^{\mu\nu} \left(y,t-\tfrac{yz}{c}\right).
\end{equation}
Here $\bar{\tau}^{\mu\nu}$ is a function both of the stress-energy tensor of the material source and of the field variable $\bar{h}^{\mu\nu}$ that is re-calculated at each iteration.   The functions $\bar{\tau}_L^{\mu\nu}$ are determined by \eqref{6.2} and \eqref{6.3} in terms of the value of $\bar{\tau}^{\mu\nu}$ outside the source.

The operator $\widetilde{\mathcal{I}^{-1}}$ is defined by \eqref{3.27}; it is the interpretation of this operator in \cite{PB:2002} that is the source of the error of that paper.  An overline denotes the expansion of the function concerned in powers of $r/c$.   Other notation is defined within the present paper and follows closely that of \cite{PB:2002}.  The form \eqref{3.34} has been used in this summary as it is valid both inside and outside the source, while \eqref{6.13} used for matching is valid only outside the source.  The result \eqref{7.4} is taken from \eqref{5.4} but expressed in terms of $\bar{\tau}_L^{\mu\nu}$ as defined by \eqref{6.2} rather than $\Lambda_L^{\mu\nu}$ as defined by \eqref{5.2}, the two forms being equivalent by \eqref{2.5}.  The function $R_L^{\mu\nu}(t)$ gives rise to the post-Newtonian expansion containing both terms logarithmic in the expansion parameter $1/c$ and tail terms that involve integration over all past time.

A companion paper will apply the results of this paper to a model problem in which logarithmic and tail terms also arise and in which the ``post-Newtonian'' expansion is known by other means \cite{Dixon:1979}.  It will be shown there that the corrected results presented here fully reproduce the expansion known from \cite{Dixon:1979} while the original results of \cite{PB:2002} give rise to a discrepancy.  It was the search for the origin of this discrepancy that led the present author to discover the error in \cite{PB:2002} that has led to the present paper.

\newpage
\renewcommand{\appendixname}{APPENDIX}
\appendix
\renewcommand{\theequation}{\Alph{section}.\arabic{equation}}

%\section{Results for STF tensors\label{STF}}
\section{RESULTS FOR STF TENSORS\label{STF}}
This appendix is provided to help to make this paper and its companion paper self-contained.  It gives outline proofs for a number of results concerning STF tensors that are used in this paper. With the exception of Lemma \ref{lem2a}, the results themselves are taken from Section A5 of Appendix A of \cite{BD:1986}, where they are given without proof.  The summation convention applies to both ordinary lower-case Latin indices and abbreviated multiple indices denoted by an upper-case Latin letter.
\begin{lem}\label{lem1}
\begin{equation}\label{A1}
  \hat{n}_L = \sum_{k=0}^{\lfloor l/2 \rfloor} a_{l,k}\,\delta_{i_1 i_2} \ldots \delta_{i_{2k-1}\, i_{2k}} n_{i_{2k+1}} \ldots n_{i_l}
\end{equation}
where
\begin{equation}\label{A2}
  a_{l,k} = (-1)^k \frac{l!}{(l-2k)!}\,\frac{(2l-2k-1)!!}{(2k)!!\,(2l-1)!!}.
\end{equation}
\begin{proof}
  By definition $\hat{n}_L$ has this form with $a_{l,0} = 1$.  Identifying the possible distinct locations for $i_1$ and $i_2$ and then contracting on these two indices gives the recurrence relation
  \begin{equation*}
    2k(2l-2k+1)a_{l,k} = - (l-2k+2)(l-2k+1)a_{l,k-1}.
  \end{equation*}
  The result follows from this relation with the initial condition $a_{l,0} = 1$.
\end{proof}
\end{lem}

\begin{lem}\label{lem2}
\begin{equation}\label{A3}
  \hat{n}_L \hat{n}'_L = n_L \hat{n}'_L = \hat{n}_L n'_L = \frac{l!}{(2l-1)!!}\,P_l(\bn\cdot\bn')
\end{equation}
where $P_l(\bn\cdot\bn')$ is the $l$th Legendre polynomial.
\end{lem}
\begin{proof}
The first two equalities hold as the traces are removed if either of the two tensors is trace-free.  If \eqref{A1} is contracted with $l$ copies of the vector $n'_i$ then the result can be put in the form
\begin{equation*}
  \hat{n}_L n'_L = \frac{1}{2^l\,(2l-1)!!}\,\frac{\dif^l}{\dif z^l}(z^2-1)^{2l}
\end{equation*}
where $z := \bn\cdot\bn'$.  The result follows on comparing this with the Rodrigues formula for Legendre polynomials.
\end{proof}

\begin{lem}\label{lem2a}
If $n \in \mathbb{N}$, $x \in \mathbb{R}$ with $-1 < x < 1$ and $h \in \mathbb{C}$ with $|h| < 1$ then
\begin{equation}\label{A3a}
  (1-2hx+h^2)^{n-\frac{\scriptscriptstyle{1}}{\scriptscriptstyle{2}}} = \sum_{l=0}^\infty P_l(x) \sum_{i=0}^n a^{(n)}_{l,i}\,h^{l+2(n-i)}
\end{equation}
where
\begin{equation}\label{A3b}
  a^{(n)}_{l,i} = (-1)^i(2l+1)\,\frac{(2n)!}{2^{2n+1}\,i!\,(n-i)!}\,\frac{\Gamma(l-i+\tfrac{1}{2})} {\Gamma(l-i+n+\tfrac{3}{2})}.
\end{equation}
\end{lem}
\begin{proof}
The result holds for $n=0$ since \eqref{A3b} gives $a^{(0)}_{l,0} = 1$ for all $l$ and so for this case, \eqref{A3a} reduces to the generating function for Legendre polynomials.  The proof is by induction in two stages from this starting point.  The first stage proves that \eqref{A3a} holds for all $n \in \mathbb{N}$ for some coefficients $a^{(n)}_{l,i}$, the second stage proves that they are given by \eqref{A3b}.

If \eqref{A3a} holds for some $n$ then multiplication by $(1-2hx+h^2)$ and use of the recurrence relation
\begin{equation*}
  (2l+1)xP_l(x) = (l+1)P_{l+1}(x) + lP_{l-1}(x)
\end{equation*}
shows that it holds for $n+1$ for some $a^{(n+1)}_{l,i}$.  Substitution of \eqref{A3a} for $n$ and $n+1$ into the identity
\begin{equation*}
  h\,\frac{\dif}{\dif h}(1-2hx+h^2)^{n+\frac{\scriptscriptstyle{1}}{\scriptscriptstyle{2}}}
  = \big(n+\tfrac{1}{2}\big)\big[(1-2hx+h^2)^{n+\frac{\scriptscriptstyle{1}}{\scriptscriptstyle{2}}}
  + (h^2-1)(1-2hx+h^2)^{n-\frac{\scriptscriptstyle{1}}{\scriptscriptstyle{2}}})\big]
\end{equation*}
and further use of the recurrence relation gives
\begin{equation*}
  \big(n+l-2i+\tfrac{3}{2}\big)\,a^{(n+1)}_{l,i} = \big(n+\tfrac{1}{2}\big) \big(a^{(n)}_{l,i-1}-a^{(n)}_{l,i}\big)
\end{equation*}
where $a^{(n)}_{l,i}$ is taken to be zero if $i<0$ or $i>n$.  This is satisfied by \eqref{A3b} identically, and as it also satisfies the initial condition $a^{(0)}_{l,0} = 1$, it holds for all values of its parameters by induction.
\end{proof}

\begin{lem}\label{lem3}
  \begin{equation}\label{A4}
    n_i \hat{n}_{iL} = \frac{l+1}{2l+1}\,\hat{n}_L.
  \end{equation}
\end{lem}
\begin{proof}
  The left hand side is an STF tensor composed of $n_i$'s and $\delta_{ij}$'s and so is a multiple of $\hat{n}_L$.  The multiplier is the coefficient of the term in its expansion by Lemma \ref{lem1} that is composed solely of $n_i$'s.  This can come only from the $k=0$ and $k=1$ terms in the sum in \eqref{A1} and is therefore $\big(1+ \tfrac{2}{l+1}a_{l+1,1})\big)$.  The result then follows from \eqref{A2}.
\end{proof}

\begin{lem}\label{lem4}
  \begin{equation}\label{A5}
    n_i\,\hat{n}_{a_1 \ldots a_l}= \hat{n}_{i a_1 \ldots a_l} + \frac{l}{2l+1}\,\delta_{i< a_1} \hat{n}_{a_2 \ldots a_l>}.
  \end{equation}
\end{lem}
\begin{proof}
In the expansion of $\hat{n}_{i a_1 \ldots a_l}$ by Lemma \ref{lem1}, the index $i$ will be either on an $n_i$ or a $\delta_{ia_k}$ for some $k$.  It will also be STF on the indices $a_1\ldots a_l$ and so will have the form
\begin{equation*}
  \hat{n}_{i a_1 \ldots a_l} = A\, n_i\, n_{<a_1\ldots a_l>} + B\, \delta_{i<a_1}\,n_{a_2\ldots a_l>}
\end{equation*}
for some $A$, $B$.  Note that there can be no $\delta$-terms within the angle brackets around the indices, as they would be removed by the taking of the trace-free part. Comparing the contributions to both sides that consist solely of $(l+1)$ copies of $n_i$ shows that $A=1$.  Contracting with a further $n_i$ and use of Lemma \ref{lem3} determines $B$.  Equation \eqref{A5} is this result in a slightly different notation.
\end{proof}

\begin{lem}\label{lem5}
  \begin{equation}\label{A6}
    r\,\partial_i \hat{n}_L = (l+1)n_i \hat{n}_L - (2l+1)\hat{n}_{iL}.
  \end{equation}
\end{lem}
\begin{proof}
Since $\partial_i r = n_i$, it follows that
\begin{equation*}
  r\,\partial_i n_{a_1\ldots a_l} = r\,\partial_i(r^{-l}x_{a_1}\ldots x_{a_l})
  = l(\delta_{i(a_1}\,n_{a_2\ldots a_l)} - n_i n_L).
\end{equation*}
Taking the STF part on $L$ and use of Lemma \ref{lem4} gives the required result.
\end{proof}

\begin{lem}\label{lem6}
  \begin{equation}\label{A7}
    \hat{\partial}_L f(r) = \hat{n}_L r^l \left(r^{-1}\frac{\partial}{\partial r}\right)^l f(r),
  \end{equation}
  \begin{equation}\label{A8}
    \hat{\partial}_L r^\lambda = \lambda(\lambda-2)\ldots (\lambda-2l+2)\hat{n}_L\,r^{\lambda-l}
  \quad\mbox{for } \lambda \in \mathbb{C}.
  \end{equation}
\end{lem}
\begin{proof}
  Equation \eqref{A7} holds for $l=1$.  Assume it holds for some $l$ with $l \geq 1$.  Then with use of Lemma \ref{lem5}
  \begin{equation*}
    \partial_i \hat{\partial}_L f(r) = (2l+1)(n_i \hat{n}_L - \hat{n}_{iL})r^{l-1}\left(r^{-1}\frac{\partial}{\partial r}\right)^l f(r)
    + n_i \hat{n}_L r^{l+1}\left(r^{-1}\frac{\partial}{\partial r}\right)^{l+1} f(r).
  \end{equation*}
  Taking the STF part of this on all $(l+1)$ indices gives \eqref{A7} for $(l+1)$, so proving it for all $l$ by induction.  Equation \eqref{A8} is a special case of \eqref{A7}.
\end{proof}

\begin{lem}\label{lem7}
  \begin{equation}\label{A9}
  \hat{\partial}_L \left(\frac{f(r)}{r}\right) = \hat{n}_L \sum_{k=0}^l a_{l,k}\, r^{k-l-1}\,\frac{\partial^k}{\partial r^k}f(r)
  \end{equation}
  where
  \begin{equation}\label{A10}
    a_{l,k} = \frac{(-1)^{l+k}(2l-k)!}{k!\,(2l-2k)!!}.
  \end{equation}
\end{lem}
\begin{proof}
It follows from \eqref{A7} that \eqref{A9} holds for coefficients $a_{l,k}$ such that
\begin{equation*}
  \left(\frac{\partial}{\partial r}\circ r^{-1}\right)^l f(r) = \sum_{k=0}^l a_{l,k}\,r^{k-2l}\,\frac{\partial^k}{\partial r^k}f(r)
\end{equation*}
if $l\geq 0$, where $\circ$ denotes the composition of the two operators, multiplication and differentiation.  It also follows that $a_{0,0} = 1$.  Applying this composite operator once more leads to the recurrence relation
\begin{equation*}
  a_{l+1,k} = a_{l,k-1} - (2l-k+1)a_{l,k}
\end{equation*}
where $a_{l,k}$ is taken as zero if $k<0$ or $k>l$.  The expression \eqref{A10} satisfies this recurrence relation and the initial condition $a_{0,0} = 1$ and so is true for all $l$ and $k$ by induction on $l$.
\end{proof}

\begin{lem}\label{lem8}
  \begin{equation}\label{A11}
  n'_Q \int \dif \Omega\, \hat{n}_Q \hat{n}_P = \frac{4\uppi p!}{(2p+1)!!}\, \delta_{pq}\,\hat{n}'_P
  \end{equation}
  where $\dif \Omega$ is an element of solid angle in the direction $\bn$.
\end{lem}
\begin{proof}
Consider the integral
\begin{equation*}
  \int \dif \Omega\, n_Q\,n_P
\end{equation*}
where the STF parts have not yet been taken.  This is an isotropic tensor of rank $l := p+q$ that is zero by symmetry if $l$ is odd. When $l$ is even it is a multiple, say $K(l)$, of the totally symmetrized product of $l/2$ Kronecker deltas.  Contract on two indices, say $i_1$ and $i_2$.  Of the $l!$ orderings of the indices on the deltas, $i_1$ and $i_2$ will be on the same delta for a fraction $1/(l-1)$ of them and will contract to give a factor $3$ while on the remaining fraction $(l-2)/(l-1)$ the contraction will result in another delta.  Hence
\begin{equation*}
  K(l-2) = \frac{l+1}{l-1}\,K(l).
\end{equation*}
Since $l$ is even and $K(0) = 4\uppi$, this iterates to give $K(l) = 4\uppi/(l+1)$.

When the STF part is taken on the $p$ indices represented by $P$, the result will be zero if any two of those indices are on the same delta.  The same holds for the $q$ indices represented by $Q$, and hence the result can only be nonzero if $p=q$, when $l=2p$.  In this case, of the $(2p)!$ possible orderings of the indices on the deltas, the indices $P$ can be paired with the deltas in $p!$ ways, as can the indices $Q$.  They can be in either order on each delta, so exactly $2^p (p!)^2$ of the $(2p)!$ orderings survive the taking of the STF parts.  Each ordering that survives replaces one of the $Q$-indices on the factor $n'_Q$ by a $P$-index.  It follows that for this case
\begin{equation*}
  n'_Q \int \dif \Omega\, \hat{n}_Q \hat{n}_P = \frac{2^p (p!)^2}{(2p)!}K(2p)\hat{n}'_P = \frac{4\uppi p!}{(2p+1)!!}\hat{n}'_P
\end{equation*}
as required.
\end{proof}

%\section{Evaluation of the generalized Poisson integral\label{PI}}
\section{EVALUATION OF THE GENERALIZED POISSON INTEGRAL\label{PI}}
This appendix shows how the operator $\triangle^{-k-1}$ defined by the generalized Poisson integral \eqref{3.7} may be applied to a source expressed as a multipole series.  It then uses this to derive \eqref{3.9} and \eqref{3.9a}.

Consider \eqref{3.7} applied to the multipole expansion
\begin{equation}\label{C0}
  \mathcal{M}(\overline{\tau})(\bx,t) = \sum_{l=0}^\infty \hat{n}_L \bar{\sigma}_L(r,t)
\end{equation}
of the source function $\overline{\tau}$.  By Lemma \ref{lem2a} of Appendix \ref{STF} and with the notation of that lemma, the term $|\bx-\by|^{2k-1}$ can be expanded in Legendre polynomials as
\begin{equation}\label{C1}
  |\bx-\by|^{2k-1} = x^{2k-1}\sum_{m=0}^\infty P_m(\bn\cdot\bn') \sum_{i=0}^k a^{(k)}_{m,i}\left(\frac{y}{x}\right)^{m+2(k-i)}
\end{equation}
for $y<x$, where $\bx = x\bn$ and $\by = y\bn'$.  The result for $y>x$ is obtained by interchanging $\bx$ and $\by$.  When these are used in \eqref{3.7}, the radial and angular integrations for each term of \eqref{C0} separate and the angular one takes the form
\begin{equation}\label{C2}
  \int\dif \Omega'\, \hat{n}'_L\,P_m(\bn\cdot\bn') = \frac{(2m-1)!!}{m!}n_M  \int\dif \Omega'\, \hat{n}'_L \hat{n}'_M = \frac{4\uppi}{2l+1}\delta_{lm}\,\hat{n}_L.
\end{equation}
Here $\dif \Omega'$ is an element of solid angle in the direction $\bn'$, the first equality is by Lemma \ref{lem2} and the second by Lemma \ref{lem8}.  It follows that
\begin{multline}\label{C3}
   \triangle^{-k-1}[\mathcal{M}(\overline{\tau})](\bx,t)\\
   = - \sum_{L=0}^\infty \frac{\hat{n}_L}{(2k)!\,(2l+1)} \sum_{i=0}^k a^{(k)}_{l,i}\bigg\{\int_0^x \dif y\,\bar{\sigma}_L(y,t)y^2 x^{2k-1}\bigg(\frac{y}{x}\bigg)^{l+2(k-i)}\\
   + \int_x^\infty \dif y\,\bar{\sigma}_L(y,t)y^2 y^{2k-1}\bigg(\frac{x}{y}\bigg)^{l+2(k-i)}  \bigg\}.
\end{multline}
This is itself a multipole expansion, showing that the multipole expansion operator $\mathcal{M}$ satisfies
\begin{equation}\label{C3a}
  \mathcal{M}(\triangle^{-k-1}[\overline{\tau}])= \triangle^{-k-1}[\mathcal{M}(\overline{\tau})].
\end{equation}

The regularization process of Section \ref{PN} begins by applying \eqref{3.7} with $\bar{\tau}$ multiplied by $\tilde{r}^B$ and then splitting the range of integration in \eqref{3.7} into the two parts $|\by|<\mathcal{R}$ and $|\by|>\mathcal{R}$.  This is equivalent to multiplying the function $\bar{\tau}(\by,t)$ respectively by $|\tilde{\by}|^B H(\mathcal{R}-|\by|)$ and $|\tilde{\by}|^B H(|\by|-\mathcal{R})$, where $H$ is the Heaviside step function.  In terms of \eqref{C3} this corresponds to multiplying $\bar{\sigma}_L(y,t)$ by $\tilde{y}^BH(\mathcal{R}-y)$ and $\tilde{y}^BH(y-\mathcal{R})$ respectively.  The results for the two parts are then analytically continued in $B$ and added to give $\triangle_B^{-k-1}[\mathcal{M}(\overline{\tau})]$.

The multipole form of \eqref{C3} is retained during these processes, so \eqref{C3a} also holds for $\triangle_B^{-k-1}$.  The final step in the regularization process is to apply the finite part operator $\FP\limits_{B=0}$.  This will act term by term on the multipole series, retaining its form, so that \eqref{C3a} holds for $\widetilde{\triangle^{-k-1}}$ as well.  With the notation of \eqref{3.27} it follows that
\begin{equation}\label{C3b}
  \mathcal{M}(\widetilde{\mathcal{I}^{-1}}[\overline{\tau}])= \widetilde{\mathcal{I}^{-1}}[\mathcal{M}(\overline{\tau})].
\end{equation}
\emph{This appears to contradict equation (3.23) of \cite{PB:2002}, but the apparent contradiction is due to a difference in the meaning of the operator $\mathcal{M}$.  The meaning in \cite{PB:2002} varies from one use of the operator to another.  Although this is explained, it is very misleading.  In the present paper all operators have a consistent meaning throughout and $\mathcal{M}$ always denotes multipole expansion.}

Now consider \eqref{C3} with $\bar{\tau}(\bx,t)= r^a$ and $a \in \mathbb{Z}$.  Apply it to the first of the two parts from the regularization process.  Due to the step function there is no convergence problem at $y=\infty$.  For sufficiently large and positive $\Re(B)$ the contribution from the limit $y=0$ will be zero.  It therefore remains zero when the result of the integration is extended by analytic continuation.  The other integral from the split is similar but in that case there is no convergence problem at $y=0$  and the contribution from the limit $y=\infty$ is zero for sufficiently large and negative $\Re(B)$.  Recombining the two parts after analytic continuation therefore gives
\begin{multline}\label{C4}
 \Delta_B^{-k-1}[\hat{n}_L r^{B+a}] = \frac{\hat{n}_L r^{B+a+2k+2}}{(2k)!\,(2l+1)}\\
 \times \sum_{i=0}^k a^{(k)}_{l,i}
 \bigg\{ \frac{1}{B+a+2i+2-l} -\frac{1}{B+a+2k-2i+3+l} \bigg\}
\end{multline}
where for simplicity the factor $r_0^{-B}$ from the regularization parameter $r_0$ has been cancelled.

The case $k=0$ of \eqref{C4} gives \eqref{3.9} as required.  That result may be iterated and the products that arise expressed in terms of $\Gamma$-functions to give
\begin{equation}\label{C5}
  (\Delta_B^{-1})^{k+1}[\hat{n}_L r^{B+a}] = \hat{n}_L r^{A+2k}f^{(k)}_l(A)
\end{equation}
where $A := B+a+2$ and
\begin{equation}\label{C6}
  f^{(k)}_l(A) := \frac{\Gamma \big((A+1+l)/2\big)\,\Gamma \big((A-l)/2\big)}{2^{2k+2}\,\Gamma \big((A+2k+3+l)/2\big)\,\Gamma \big((A+2k+2-l)/2\big)}.
\end{equation}
Alternatively the products can be expressed in terms of linear partial fractions.  It can be seen from \eqref{C4} and \eqref{C5} that the denominators of the partial fractions are precisely those that occur in \eqref{C4} for the same value of $k$, so that
\begin{equation}\label{C7}
  f^{(k)}_l(A) = \sum_{i=0}^k \bigg\{ \frac{b^{(k)}_{l,i}}{A+2i-l} + \frac{c^{(k)}_{l,i}}{A+2k-2i+1+l} \bigg\}
\end{equation}
where $b^{(k)}_{l,i}$ and $c^{(k)}_{l,i}$ are given by
\begin{equation}\label{C8}
\left.
\begin{array}{c}
  b^{(k)}_{l,i} = \lim\limits_{A \to -2i+l} (A+2i-l)f^{(k)}_l(A), \vspace{6pt} \\ c^{(k)}_{l,i} = \lim\limits_{A \to -2k+2i-1-l} (A+2k-2i+1+l)f^{(k)}_l(A).
\end{array}
\right\}
\end{equation}

For each $\Gamma$-function in \eqref{C6} the limit is given just by taking $A$ to be the limiting value, except where this gives an argument that is a negative integer.  In this case it needs to be transformed by the reflection formula ${\Gamma(z)\Gamma(1-z) = \uppi / \sin(\uppi z)}$.  In this way it is found that
\begin{equation}\label{C9}
   b^{(k)}_{l,i} = -c^{(k)}_{l,i} = \frac{1}{(2k)!\,(2l+1)}\,a^{(k)}_{l,i}
\end{equation}
where $a^{(k)}_{l,i}$ is given by \eqref{A2}.  Putting this back into \eqref{C7} and the result into \eqref{C5} shows, by comparison with \eqref{C4}, that
\begin{equation}\label{C10}
  \Delta_B^{-k-1}[\hat{n}_L r^{B+a}] = (\Delta_B^{-1})^{k+1}[\hat{n}_L r^{B+a}]
\end{equation}
as required.

Finally there is an important extension of this result.  The functional $\Delta_B^{-k-1}$ commutes with $\partial/\partial B$, since this commutes with the separate integrations in the two parts of the split integral for the ranges of $\Re(B)$ for which they converge and differentiation commutes with analytic continuation.  It follows by repeated differentiation with respect to $B$ that
\begin{multline}\label{C11}
  \Delta_B^{-k-1}[\hat{n}_L r^{B+a}(\log r)^p] = (\Delta_B^{-1})^{k+1}[\hat{n}_L r^{B+a}(\log r)^p] \\
  = \frac{\partial^p}{\partial B^p}\Delta_B^{-k-1}[\hat{n}_L r^{B+a}]
  = \frac{\partial^p}{\partial B^p}(\Delta_B^{-1})^{k+1}[\hat{n}_L r^{B+a}].
\end{multline}

%\section{Post-Newtonian expansion of the external solution\label{PNOS}}
\section{POST-NEWTONIAN EXPANSION OF THE EXTERNAL SOLUTION\label{PNOS}}
This appendix provides a direct proof of \eqref{6.14}, to remove any possible doubt as to which version of the operator of instantaneous potentials $\widetilde{\mathcal{I}^{-1}}$ it involves.  Substituting \eqref{5.3} into \eqref{5.6} and expanding the result in powers of $(1/c)$ gives
\begin{equation}\label{B1}
  \frac{16\uppi G}{c^4}\,\overline{S_1^{\mu\nu}}(\bx,t) := \FP\limits_{B=0} \sum_{n=0}^\infty \left(\frac{\partial}{c\,\partial t}\right)^n \left(r_0^{-B} \overline{\underset{n}{S}{}_B^{\mu\nu}}(\bx,t)\right)
\end{equation}
where
\begin{equation}\label{B2}
  \overline{\underset{n}{S}{}_B^{\mu\nu}}(\bx,t) := \int_{-r}^{r} \dif s\,\sum_{l=0}^\infty \hat{\partial}_L \Bigg\{ \frac{(s+r)^l}{2r} \int_0^{\tfrac{s+r}{2}} \dif y\, \frac{(\tfrac{s+r}{2}-y)^l}{l!}\, y^{B-l+1}\frac{(y-s)^n}{n!}\overline{\Lambda_L^{\mu\nu}}(y,t)\Bigg\}
\end{equation}
and as before, an overline denotes post-Newtonian expansion.  This is a relationship between the post-Newtonian expansions of $S_1^{\mu\nu}$ and $\Lambda^{\mu\nu}$.  Note from \eqref{5.3} and \eqref{5.6} that $S_1^{\mu\nu}(\bx,t)$ for a fixed $(\bx,t)$ depends on $\Lambda^{\mu\nu}(\bx',t')$ only for $|\bx'|<|\bx|$ and $|t-t'|<|\bx|/c$ so that nesting post-Newtonian expansions in this way is valid.  Moreover, $\Lambda^{\mu\nu}$ in this region depends only on $h^{\mu\nu}$ within this same region.

On expanding the action of $\hat{\partial}_L$ by \eqref{A9} and putting $y = \tfrac{s+r}{2}z$, this becomes
\begin{multline}\label{B3}
  \overline{\underset{n}{S}{}_B^{\mu\nu}}(\bx,t) = \int_{-r}^{r} \dif s\,\sum_{l=0}^\infty \hat{n}_L
  \sum_{k=0}^l \frac{(-1)^{k+l}(2l-k)!}{k!\,(2l-2k)!!}\,r^{k-l-1} \frac{\partial^k}{\partial r^k}\Bigg\{ \frac{(s+r)^{B+l+2}}{2^{B+3}}\\ \times \int_0^1 \dif z\, \frac{(1-z)^l}{l!\,n!}  z^{B-l+1}\left(\tfrac{s+r}{2}z-s\right)^n \overline{\Lambda_L^{\mu\nu}}(\tfrac{s+r}{2}z,t)\Bigg\}.
\end{multline}
The region dependence described above shows that the asymptotic form of $h^{\mu\nu}$ as $r \to 0$ can be determined by iterating \eqref{B3}.  There are other contributions to $h^{\mu\nu}$ from $X_L^{\mu\nu}$ via \eqref{5.9} and $R_L^{\mu\nu}$ via \eqref{5.7} but their asymptotic forms are known.  They can be expanded about $r=0$ as power series in $r/c$ and their contributions to $\mathcal{M}(h_\ext^{\mu\nu})$ will be terms of the form $f_L^{\mu\nu}(t) r^a$ where $a \in \mathbb{Z}$ with $a \geq -l-1$.  Their contributions to $\Lambda_L^{\mu\nu}$ for a given $l$ will therefore have the same form with $a$ bounded below at any post-Newtonian order $n$ with some lower bound $-N(n)$ for $a$, though $N \rightarrow \infty$ as $n \rightarrow \infty$.

The first iteration of \eqref{B3} will have no other terms, so for now consider one such individual term for one value of $l$, say $\Lambda_L^{\mu\nu} = f_L^{\mu\nu}(t) r^a$.  Put this into \eqref{B3} and expand the bracket with power $n$ by the binomial theorem to give
\begin{multline}\label{B4}
  \overline{\underset{n}{S}{}_B^{\mu\nu}}(\bx,t) = \hat{n}_L f_L^{\mu\nu}(t)  \sum_{k=0}^l \sum_{i=0}^n \frac{1}{l!\,n!} \binom{l}{k} \binom{n}{i}r^l \frac{\partial^{l-k}r^{-l-1}}{\partial r^{l-k}}\\ \times \left\{ \int_{-r}^{r} \dif s\,(-s)^i \frac{\partial^k}{\partial r^k}\left( \frac{(s+r)^{A+l+n-i}}{2^{A+1+n-i+l-k}}\right)\right\}\\ \times \left\{ \int_0^1 \dif z\, (1-z)^l  z^{A-1-l+n-i}\right\}
\end{multline}
where $A:=B+a+2$.  The factor $r^{k-l-1}$ in \eqref{B3} has been converted into a derivative term for later convenience.  The two integrals can now be performed by repeated integration by parts to give
\begin{multline}\label{B5}
  \int_{-r}^{r} \dif s\,\left\{(-s)^i \frac{\partial^k}{\partial r^k}\left( \frac{(s+r)^{A+l+n-i}}{2^{A+1+n-i+l-k}}\right)\right\} \\
  = \sum_{j=0}^i \frac{2^j\,i!}{(i-j)!}(-r)^{i-j} \frac{\partial^k}{\partial r^k}\left(\frac{\Gamma(A+l+n-i+1)}{\Gamma(A+l+n-i+j+2)}r^{A+l+n-i+j+1} \right)
\end{multline}
and
\begin{equation}\label{B6}
  \int_0^1 \dif z\, (1-z)^l  z^{A-1-l+n-i} = \frac{l!\,\Gamma(A+n-i-l)}{\Gamma(A+n-i+1)}.
\end{equation}
Gamma functions are used here instead of factorials since $A \in \mathbb{C}$.  Note that in \eqref{B5} the terms come only from the upper limit of the integral and that the $k$-fold differentiation with respect to $r$ is logically performed \emph{before} the substitution $s=r$.   To achieve the same result with the differentiation performed \emph{after} substitution, as in \eqref{B5}, an extra factor $2^k$ has been placed in the denominator.

When these are substituted back into \eqref{B4}, the surviving terms involving $k$ form the expansion of a $k$-fold derivative by the Leibnitz formula.  Use this to perform first the summation over $k$ and then the $k$-fold derivative itself.  In the result, change the summation variable from $i$ to $m$ where $m = n-i+j$, so that $\sum_{i=0}^n\sum_{j=0}^i$ becomes $\sum_{m=0}^n\sum_{j=0}^m$.  This gives
\begin{equation}\label{B7}
  \overline{\underset{n}{S}{}_B^{\mu\nu}}(\bx,t) = \sum_{m=0}^n\sum_{j=0}^m c_{n,m,j}(A,l) \hat{n}_L r^{A+n} f_L^{\mu\nu}(t)
\end{equation}
where
\begin{multline}\label{B8}
  c_{n,m,j}(A,l) :=  \frac{(-1)^{n-m}\, 2^j}{(m-j)!\,(n-m)!}\,\frac{ \Gamma(A-l+m-j)}{\Gamma(A-l+m+1)} \\ \times \frac{\Gamma(A+l+m-j+1)}{\Gamma(A+l+m+2)}\,\frac{\Gamma(A+m+1)}{\Gamma(A+m+1-j)}.
\end{multline}

In this form $c_{n,m,j}(A,l)$ is defined for any $l \in \mathbb{C}$, not just $l \in \mathbb{N}$.  By manipulating it in this form, it avoids the need to treat special cases as the result can be extended to all $l$ by analytic continuation even if steps in the derivation are not valid for certain integer values.  Now considered as a function of $A \in \mathbb{C}$, each quotient in \eqref{B8} is a polynomial expressed as a product of linear factors, the first two being of degree $(j+1)$ on the denominator and the third being of degree $j$ on the numerator.  For a general $l \in \mathbb{C}$ their linear factors are all distinct, so that $c_{n,m,j}(A,l)$ can be expanded in partial fractions of the form
\begin{equation}\label{B9}
   c_{n,m,j}(A,l) = \sum_{k=m-j}^m \left\{ \frac{a_{n,m,j,k}(l)}{A-l+k} + \frac{b_{n,m,j,k}(l)}{A+l+1+k} \right\}.
\end{equation}
The coefficients are given by
\begin{equation}\label{B10}
\left.
\begin{array}{c}
  a_{n,m,j,k}(l) = \lim\limits_{A \to l-k} (A-l+k)c_{n,m,j}(A,l), \vspace{6pt} \\ b_{n,m,j,k}(l) = \lim\limits_{A \to -l-1-k} (A+l+1+k)c_{n,m,j}(A,l).
\end{array}
\right\}
\end{equation}
For each $\Gamma$-function in \eqref{B8} the limit is given just by taking $A$ to be the limiting value, except where this gives an argument that is a negative integer.  In this case it needs to be transformed by the reflection formula ${\Gamma(z)\Gamma(1-z) = \uppi / \sin(\uppi z)}$.  In this way it is found that
\begin{multline}\label{B11}
  a_{n,m,j,k}(l) = \frac{(-1)^{n-j+k}\, 2^j}{(m-j)!\,(n-m)!\, (k-m+j)!\,(m-k)!}\\ \times \frac{\Gamma(2l+m-j-k+1)}{\Gamma(2l+m-k+2)}\,\frac{\Gamma(l+m-k+1)}{\Gamma(l+m-k+1-j)}
\end{multline}
and
\begin{equation}\label{B12}
  b_{n,m,j,k}(l) = (-1)^{n+1}a_{n,n-m+j,j,n-k}.
\end{equation}
\begin{sloppypar}
Now put \eqref{B9} back into \eqref{B7}, and reorder the triple sum from ${\sum_{m=0}^n \sum_{j=0}^m \sum_{k=m-j}^m}$ to ${\sum_{k=0}^n \sum_{m=k}^n \sum_{j=m-k}^m}$.  Separate the sums over the $a$-terms and $b$-terms and change the variables from $(k,m,j)$ to $(p,q,r)$ with different transformations for each sum.  For the $a$-terms take $p=k$, $q=m-k$, $r=j-m+k$ and for the $b$-terms take $p=n-k$, $q=j-m+k$, $r=m-k$.  With use of \eqref{B12} the result takes the form
\begin{equation}\label{B13}
  \overline{\underset{n}{S}{}_B^{\mu\nu}}(\bx,t) = \hat{n}_L r^{A+n} f_L^{\mu\nu}(t) \sum_{p=0}^n c_{n,p}(l)\left\{ \frac{1}{A-l+p} - \frac{(-1)^n}{A+l+1+n-p} \right\}
\end{equation}
where
\begin{equation}\label{B14}
  c_{n,p}(l) := \sum_{q=0}^{n-p} \sum_{r=0}^p a_{n,p+q,q+r,p}.
\end{equation}
\end{sloppypar}

If \eqref{B11} is used in \eqref{B14}, the result factorizes as
\begin{multline}\label{B15}
  c_{n,p}(l) = \frac{(-1)^{n-p}}{p!\,(n-p)!} \bigg\{\sum_{q=0}^{n-p} \binom{n-p}{q}\frac{(-2)^q \Gamma(l+q+1)}{\Gamma(2l+q+2)} \bigg\} \\ \times \bigg\{\sum_{r=0}^{p} \binom{p}{r}\frac{(-2)^r \Gamma(2l-r)}{\Gamma(l+1-r)} \bigg\}.
\end{multline}
Recall now that here $l \in \mathbb{C}$ is a generic value chosen to avoid special cases.  The reflection formula may therefore be used on both $\Gamma$-functions in the sum over $q$ to put \eqref{B15} in the form
\begin{equation}\label{B16}
  c_{n,p}(l) = (-1)^{n-p+1}\,\frac{2\cos(l\uppi)}{p!\,(n-p)!}f(n-p,-l-1)\,f(p,l)
\end{equation}
where
\begin{equation}\label{B17}
  f(p,l) := \sum_{r=0}^{p} \binom{p}{r}\frac{(-2)^r\, \Gamma(2l-r)}{\Gamma(l+1-r)}.
\end{equation}

Consider now the formal Taylor series expansion of
\begin{equation}\label{B18}
  g(x,y) := (1+y)^{2l} \left( 1-\frac{2x}{1+y}\right)^l.
\end{equation}
This is
\begin{equation}\label{B19}
  g(x,y) = \sum_{i=0}^\infty \frac{\Gamma(l+1)}{i!\,\Gamma(l-i+1)}\,(-2x)^i \bigg( \sum_{j=0}^\infty \frac{\Gamma(2l-i+1)}{j!\,\Gamma(2l-i-j-1)}\,y^j \bigg).
\end{equation}
On setting $y=x$, the series can be rearranged as
\begin{equation}\label{B20}
  g(x,x) := \sum_{p=0}^\infty \frac{\Gamma(l+1)}{p!\,\Gamma(2l-p+1)}\,f(p,l)x^p.
\end{equation}
But
\begin{equation}\label{B21}
  g(x,x) \equiv (1-x^2)^l = \sum_{p=0}^\infty (-1)^p\,\frac{\Gamma(l+1)}{p!\,\Gamma(l-p+1)}\,x^{2p}.
\end{equation}
Equating the two series shows that
\begin{equation}\label{B22}
\left.
\begin{array}{lll}
  f(2p,l)& = &\displaystyle{(-1)^p\,\frac{(2p)!}{p!}\,\frac{\Gamma(2l-2p+1)}{\Gamma(l-p+1)}}\vspace{3pt}\\
  f(2p+1,l)& =& 0
\end{array}
\right\}
\end{equation}
for $p \in \mathbb{N}$.  With the aid of the duplication formula
\begin{equation}\label{B23}
  \Gamma(2z) = \uppi^{-1/2}\,2^{2z-1} \Gamma(z)\Gamma(z+\tfrac{1}{2})
\end{equation}
and further use of the reflection formula, the first of these can be expressed in either of the forms
\begin{equation}\label{B24}
  f(2p,l) = 2^{2l-2p}\,\frac{(2p)!}{p!} \times \left\{
\begin{array}{l}
  (-1)^p\, \uppi^{-1/2}\, \Gamma(l-p+\tfrac{1}{2}) \\
  \mbox{or }-\uppi^{1/2}/\big(\cos(l\uppi) \Gamma(p-l+\tfrac{1}{2}) \big).
\end{array}
\right.
\end{equation}

If \eqref{B24} is used with \eqref{B16} it shows that $c_{n,p}(l)=0$ if either $n$ or $p$ is odd, and that
\begin{equation}\label{B25}
  c_{2n,2p} = \frac{(-1)^p}{2^{2n+1}\,p!\,(n-p)!}\,\frac{\Gamma(l-p+\tfrac{1}{2})}{\Gamma(l+n-p+\tfrac{3}{2})}.
\end{equation}
This brings \eqref{B13} to the form
\begin{multline}\label{B26}
  \overline{\underset{n}{S}{}_B^{\mu\nu}}(\bx,t) = \hat{n}_L r^{B+a+2+n} f_L^{\mu\nu}(t) \sum_{p=0}^n \frac{(-1)^p}{2^{2n+1}\,p!\,(n-p)!}\,\frac{\Gamma(l-p+\tfrac{1}{2})}{\Gamma(l+n-p+\tfrac{3}{2})} \\ \times \left\{ \frac{1}{B+a+2+2p-l} - \frac{1}{B+a+3+2n-2p+l} \right\}.
\end{multline}
But by \eqref{C4} with \eqref{A3b} this is precisely the partial fractions expression for
\begin{equation}\label{B27}
  \overline{\underset{n}{S}{}_B^{\mu\nu}}(\bx,t) = \triangle_B^{-n-1}[\hat{n}_L r^{B+a} f_L^{\mu\nu}(t)].
\end{equation}
Putting this back into \eqref{B1} and making use of the definition \eqref{3.27} gives
\begin{equation}\label{B28}
  \frac{16\uppi G}{c^4}\,\overline{S_1^{\mu\nu}}(\bx,t) = \widetilde{\mathcal{I}^{-1}}[\hat{n}_L \overline{\Lambda_L^{\mu\nu}}]
\end{equation}
for the case $\overline{\Lambda_L^{\mu\nu}}(r,t) = r^a f_L^{\mu\nu}(t)$ that is being considered.

The same result will hold when $\overline{\Lambda_L^{\mu\nu}}$ is a linear combination of such terms.  This is not the most general form for $\overline{\Lambda_L^{\mu\nu}}$, however.  As shown by \eqref{3.15}, such a $\overline{\Lambda_L^{\mu\nu}}$ can lead to terms in $\overline{S_1^{\mu\nu}}$, and therefore in $\mathcal{M}(h_\ext^{\mu\nu})$, of the form $r^a (\log r) f_L^{\mu\nu}(t)$.  These in turn may give rise to contributions to $\overline{\Lambda_L^{\mu\nu}}$ for the next iteration that have the form $r^a (\log r)^p\,f_L^{\mu\nu}(t)$ for $a \in \mathbb{Z}$, $p \in \mathbb{N}$.  Differentiation of \eqref{B27} repeatedly with respect to $B$ gives the corresponding result for $\overline{\Lambda_L^{\mu\nu}}(r,t) = r^a (\log r)^p f_L^{\mu\nu}(t)$ so that \eqref{B28} holds also for such terms.

The finite part operation of the $\widetilde{\mathcal{I}^{-1}}$ operator leaves the general form of such terms unchanged, at most increasing the value of $p$ by one, so the value of $\mathcal{M} (h_\ext^{\mu\nu})$ from this iteration is also composed of terms of this form.  No further type of term arises in the construction of the next iteration for $\overline{\Lambda_L^{\mu\nu}}$ from such $\overline{h_L^{\mu\nu}}$.  It therefore follows that
\begin{equation}\label{B29}
  \overline{\mathcal{M}(h_\ext^{\mu\nu})}(\bx,t) = \sum_{l,a,p} \hat{n}_L r^a (\log r)^p\,F_{L,a,p}^{\mu\nu}(t)
\end{equation}
and
\begin{equation}\label{B30}
  \overline{\mathcal{M}(\Lambda_\ext^{\mu\nu})}(\bx,t) = \sum_{l,a,p} \hat{n}_L r^a (\log r)^p\,G_{L,a,p}^{\mu\nu}(t)
\end{equation}
for some smooth functions $F_{L,a,p}^{\mu\nu}(t)$ and $G_{L,a,p}^{\mu\nu}(t)$ with $l,p \in \mathbb{N}$ and $a \in \mathbb{Z}$ and that
\begin{equation}\label{B31}
  \frac{16\uppi G}{c^4}\,\overline{S_1^{\mu\nu}} = \widetilde{\mathcal{I}^{-1}}[\overline{\mathcal{M}(\Lambda_\ext^{\mu\nu})}]
\end{equation}
as required.  The terms in \eqref{B29} satisfy $a \geq - N(n,l)$ and $p \leq P(n,l)$ at any post-Newtonian order $n$ but $N,P \rightarrow \infty$ as $n \rightarrow \infty$.  The same holds for \eqref{B30} but with different values for $N,P$.

\bibliography{wgdixon_BlanchetError}
\end{document}